\documentclass[sigconf,10pt]{acmart}

\newif\ifACM
\ACMtrue

\usepackage[T1]{fontenc}
\usepackage[utf8]{inputenc}

\usepackage{upquote}

\usepackage{microtype}
\UseMicrotypeSet[protrusion]{basicmath} 

\usepackage{environ}
\usepackage{graphicx,grffile}
\usepackage{multirow}
\usepackage{xspace}
\usepackage{tabularx}
\usepackage{ragged2e}
\usepackage{booktabs}
\usepackage{paralist}
\usepackage[american]{babel}
\usepackage[shortlabels]{enumitem}
\setlist{nosep,leftmargin=*}
\usepackage{courier,color,wrapfig}
\usepackage{balance}
\usepackage{epstopdf}
\usepackage{url}
\usepackage{pifont}
\usepackage{subfigure}
\usepackage[table]{xcolor}

\usepackage{mathtools}
\usepackage{thmtools}
\usepackage[normalem]{ulem}

\usepackage[linesnumbered]{algorithm2e}
\usepackage{etoolbox}
\usepackage[textsize=small]{todonotes}
\usepackage[compact,small]{titlesec}
\usepackage{caption}
\usepackage[export]{adjustbox}

\usepackage[capitalize]{cleveref}
\crefformat{section}{§#2#1#3}

\usepackage{tikz}
\usetikzlibrary[shapes.misc, shapes.geometric, shapes.arrows, positioning, decorations.pathmorphing, matrix, graphs, fit, calc, arrows, arrows.meta, backgrounds]

\usepackage{tcolorbox}
\usepackage{csvsimple}

\makeatletter
\def\maxwidth{\ifdim\Gin@nat@width>\linewidth\linewidth\else\Gin@nat@width\fi}
\def\maxheight{\ifdim\Gin@nat@height>\textheight\textheight\else\Gin@nat@height\fi}
\makeatother
\setkeys{Gin}{width=\maxwidth,height=\maxheight,keepaspectratio}
\makeatletter
\g@addto@macro{\UrlBreaks}{\UrlOrds}
\makeatother

\makeatletter
\let\origsection\section
\let\origsubsection\subsection

\renewcommand\section{\@ifstar{\starsection}{\nostarsection}}
\renewcommand\subsection{\@ifstar{\starsubsection}{\nostarsubsection}}

\newcommand\sectionprelude{\vspace{0ex}}
\newcommand\sectionpostlude{\vspace{0ex}}
\newcommand\subsectionprelude{\vspace{0ex}}
\newcommand\subsectionpostlude{\vspace{0ex}}

\newcommand\nostarsection[1]{\sectionprelude\origsection{#1}\sectionpostlude}
\newcommand\starsection[1]{\sectionprelude\origsection*{#1}\sectionpostlude}

\newcommand\nostarsubsection[1]{\subsectionprelude\origsubsection{#1}\subsectionpostlude}
\newcommand\starsubsection[1]{\subsectionprelude\origsubsection*{#1}\subsectionpostlude}

\makeatother

\newcommand\paraspace{\vspace*{0.25ex}}
\providecommand\parab[1]{\paraspace\noindent\textbf{#1}}

\apptocmd\normalsize{%
\abovedisplayskip=5pt
\abovedisplayshortskip=5pt
\belowdisplayskip=5pt
\belowdisplayshortskip=5pt
}{}{}

\newcommand{\sysname}{OnlineTE\xspace}

\newcommand{\xvar}{$x$}
\newcommand{\xivar}{$x_{inner}$}
\newcommand{\xhvar}{$x_{hier}$}

\newcommand{\divar}{dual$_{inner}$}
\newcommand{\zvar}{$z$}
\newcommand{\zivar}{$z_{inner}$}
\newcommand{\zhvar}{$z_{hier}$}
\newcommand{\xupd}{\xvar{}-update}
\newcommand{\xiupd}{\xivar{}-update}
\newcommand{\xhupd}{\xhvar{}-update}
\newcommand{\zupd}{\zvar{}-update}
\newcommand{\ziupd}{\zivar{}-update}
\newcommand{\zhupd}{\zhvar{}-update}
\newcommand{\dupd}{dual-update}

\newcommand{\ie}{\emph{i.e.,}\xspace}
\newcommand{\eg}{\emph{e.g.,}\xspace}

\tcbuselibrary{skins, breakable}
\newtcolorbox{outline}[1][]
{
  enhanced,
  breakable,
  skin first=enhanced,
  skin middle=enhanced,
  skin last=enhanced,
  colback=blue!5!white, 
  colframe=blue!50!black, 
  fonttitle=\bfseries, 
  boxrule=0.5pt,
  arc=4mm, 
  left=5pt,
  right=5pt,
  title={Outline},
  #1
}
\tcbset{before upper=\setlength{\parskip}{1ex}} 

\DeclareMathOperator*{\argmin}{arg\,min}

\NewDocumentCommand{\evalat}{sO{\Big}mm}{%
  \IfBooleanTF{#1}
   {\mleft. #3 \mright|_{#4}}
   {#3#2|_{#4}}%
}

\newcommand{\atiter}[2]{#1^{(#2)}}
\newcommand{\indexedatiter}[3]{#1_{#2}^{(#3)}}

\newcommand{\edgeout}[1]{$E^{\text{out}}_{#1}$}
\newcommand{\edgein}[1]{$E^{\text{in}}_{#1}$}

\definecolor{tablegray}{gray}{0.9}

\declaretheorem[name=Definition]{definition}
\declaretheorem[name=Theorem, numberlike=definition]{theorem}
\crefname{theorem}{Theorem}{Theorems}
\crefname{definition}{Definition}{Definitions}

\begin{document}

\title{Near-optimal Online Traffic Engineering}

\author{Arvin Ghavidel}
\affiliation{
    \institution{University of Southern California}
}

\author{Pooria Namyar}
\affiliation{
    \institution{Microsoft Research}
}

\author{Nikolai Matni}
\affiliation{
    \institution{University of Pennsylvania}
}

\author{Walter Willinger}
\affiliation{
    \institution{Northwestern University}
}

\author{Ramesh Govindan}
\affiliation{
    \institution{University of Southern California}
}

\date{}



\ifACM
\renewcommand\footnotetextcopyrightpermission[1]{} 
\setcopyright{none}
\settopmatter{printacmref=false, printccs=false, printfolios=true}
\acmDOI{}
\acmISBN{}
\acmConference[Submitted for review]{}
\acmYear{2018}
\acmPrice{}
\pagestyle{plain}
\fi

\ifACM


\fi

\newcommand{\grantack}[1]{
  \ifthenelse{\equal{#1}{CTA}}{\thanks{Research was sponsored by the Army Research Laboratory and was accomplished under Cooperative Agreement Number W911NF-09-2-0053 (the ARL Network Science CTA). The views and conclusions contained in this document are those of the authors and should not be interpreted as representing the official policies, either expressed or implied, of the Army Research Laboratory or the U.S. Government. The U.S. Government is authorized to reproduce and distribute reprints for Government purposes notwithstanding any copyright notation here on.}
    }{}
  \ifthenelse{\equal{#1}{CRA}}{\thanks{Research reported in this paper was sponsored in part by the Army Research Laboratory under Cooperative Agreement W911NF-17-2-0196. The views and conclusions contained in this document are those of the authors and should not be interpreted as representing the official policies, either expressed or implied, of the Army Research Laboratory or the U.S. Government. The U.S. Government is authorized to reproduce and distribute reprints for Government purposes notwithstanding any copyright notation here on.}}{}
  \ifthenelse{\equal{#1}{Conix}}{\thanks{This work was supported in part by the CONIX Research Center, one of six centers in JUMP, a Semiconductor Research Corporation (SRC) program sponsored by DARPA.}}{}
  \ifthenelse{\equal{#1}{NSFAvail}}{\thanks{This material is based upon work supported by the National Science Foundation under Grant No. 1705086}}{}
  \ifthenelse{\equal{#1}{Conix}}{\thanks{This work was supported in part by the CONIX Research Center, one of six centers in JUMP, a Semiconductor Research Corporation (SRC) program sponsored by DARPA.}}{}
  \ifthenelse{\equal{#1}{CPSSyn}}{\thanks{This material is based upon work supported by the National Science Foundation under Grant No. 1330118 and from a grant from General Motors.}}{}        
  \ifthenelse{\equal{#1}{NeTSLarge}}{\thanks{This material is based upon work supported by the National Science Foundation under Grant No. 1413978}}{}         \ifthenelse{\equal{#1}{NeTSSmall}}{\thanks{This material is based upon work supported by the National Science Foundation under Grant No. 1423505}}{}
}


\maketitle

\noindent\textbf{Abstract:}
%
Most deployed WAN Traffic Engineering (TE) systems use a logically centralized controller that periodically gathers traffic demands, runs a TE optimization or heuristic, and then programs the network.
At scale, these solutions can be sub-optimal, and can take minutes to react to demand changes or failures.
In this paper, we introduce \sysname{}, a system that reacts immediately to demand changes and failures, and delivers near-optimal solutions within seconds of a change.
\sysname{} builds on the theory of optimization decomposition to devise scalable, near-optimal, distributed TE solvers for path-based MLU and Max-flow problems.
In \sysname{}, each switch solves part of the optimization, and a central coordinator orchestrates the progress of the switches.
As such, a switch can trigger a re-optimization as soon as it notices a demand change or failure, enabling high reactivity.
\sysname{} scales to large WANs, and its compute requirements are well below the capabilities of modern WAN switches.
It also enables a new opportunity, edge-based TE, which can utilize resources more efficiently than today's path-based approaches.
On a testbed emulation of a 750-node WAN topology, \sysname{} can outperform the state-of-the-art by up to an order of magnitude.

\section{Introduction}
\label{s:introduction}

Traffic engineering (TE) is an indispensable component of today's wide area networks (WANs).
It enables cloud providers and ISPs to route \textit{demands} between each ingress-egress router pair efficiently and fairly, while ensuring good application perceived performance.
At the core of today's TE systems is an algorithm that routes demands across multiple pre-computed paths. It aims to achieve a global objective, such as maximizing total flow (henceforth, Max-flow) or minimizing maximum link utilization (henceforth MLU) without violating link capacity constraints.
TE is also used in cloud provider private WANs~\cite{krishnaswamy23:_onewan,denis23:_ebb,b4andafter,b4} and in data centers~\cite{poutievski22:_jupit}.
In these settings, TE algorithms must support additional features, such as traffic prioritization and fairness~\cite{swan,b4,b4andafter,krishnaswamy23:_onewan,denis23:_ebb}, or be designed to co-exist with other techniques like topology engineering~\cite{poutievski22:_jupit}.

\subsection{Centralized TE}\label{sec:centralized-te}
Most deployed TE algorithms run on a logically centralized controller (\cref{fig:te_loops}).
The controller periodically (\textit{e.g.}, every 5 minutes~\cite{swan,soroush,krishnaswamy23:_onewan}) (a) collects demands from switches and predicts future demand, (b) runs the TE algorithm on predicted demand, and (c) installs forwarding rules on switches based on the TE algorithm's output.
These rules \textit{split} the demand between each ingress-egress pair across pre-computed paths between those routers.
%
%
%
With a global view of traffic and topology, centralized TE can make routing decisions across the entire network and utilize network capacity efficiently.
However, there are two limitations.

\begin{figure}[t]
  \centering\scalebox{0.9}{
  \begin{tikzpicture}
    \node (A) {\includegraphics[width=1.6in,height=1.2in]{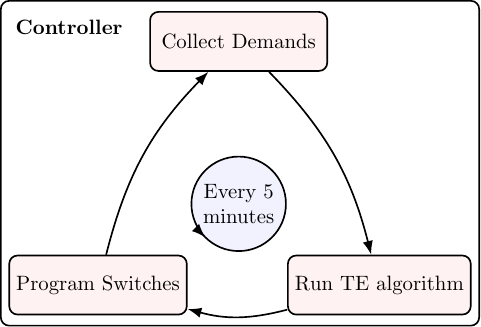}};
    \node[right=3mm of A] (B) {\includegraphics[width=1.8in,height=1.7in]{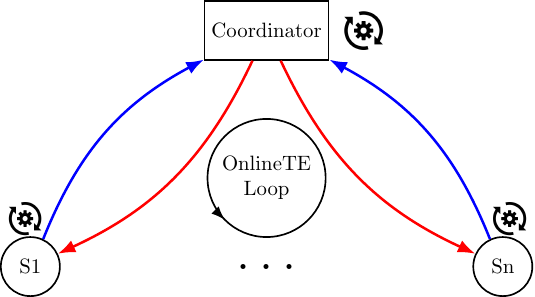}};
  \end{tikzpicture}}
  \caption{The classical TE loop (left) runs periodically on a controller. In \sysname{}'s distributed TE approach (right), switches notice demand shifts and failures immediately, and iteratively run a decomposed optimization together with a coordinator.\label{fig:te_loops} }
\end{figure}

\parab{They trade off optimality for scale.}
Optimal algorithms for allocating demands to paths have been widely studied~\cite{danna2012upward,danna2012practical}.
However, as networks scale, the time to compute these optimal allocations can increase substantially~\cite{soroush,ncflow}.
This can delay the response to demand changes or failures, and significantly impact availability.
To avoid this, most deployed TE algorithms use heuristics to scale~\cite{blastshield,krishnaswamy23:_onewan,denis23:_ebb,b4,b4andafter} which can exhibit poor worst-case behavior~\cite{metaopt}.
Some~\cite{swan,soroush} use approximation algorithms with bounded optimality gaps.

In response to this, a line of work has explored speeding up TE algorithms by contracting the network~\cite{ncflow}, partitioning~\cite{pop}, training ML models~\cite{perry23:_dote,teal_2023}, or decomposing the optimization~\cite{xu2025decoupledecomposescalingresource}.
Many of these achieve speed at the expense of optimality, and ML-based methods~\cite{perry23:_dote} can also exhibit poor worst-case behavior~\cite{learning-enabled}.

\parab{They are slow to react.}
Even with scalable TE algorithms, the overall TE control loop (\cref{fig:te_loops}) incurs non-trivial end-to-end latency.
As a result, deployed TE controllers typically re-optimize only once every few minutes~\cite{teal_2023,redTE}.
In today's networks, this can reduce availability.
Traffic demands can change over shorter timescales~\cite{soroush,redTE} than the periodicity of TE.
If these changes produce a traffic matrix outside the predicted demand used to compute the TE solution, the increased demands can overload links and cause packet drops.
Reaction to failures can also impact availability: when one or more links go down, the residual capacity may be insufficient to accommodate existing demands, so congestion might persist until the next iteration of the control loop.

%


\vspace*{-0.5ex}
\subsection{Distributed TE}\label{sec:distributed-te}

In response to these shortcomings, other approaches have considered \textit{distributed} approaches to TE.

TeXCP~\cite{kandulate} and MATE~\cite{mate} have explored an approach in which each ingress router repeatedly (a) measures congestion on demand paths, (b)  locally and independently adapts demand splits based on these measurements.
While these approaches guarantee stability, they can converge slowly and can be sub-optimal because routers do not explicitly coordinate when determining the local splits, but do so implicitly via path congestion measurements.
%

%
%

More recent work, dSDN~\cite{sigcomm2024:krentselskmnras24}, explores a different point in the design space, albeit in the cloud provider setting.
It is motivated by controller complexity, and is predicated on the availability of non-trivial compute on today's switches.
In dSDN, routers flood link state and demand changes, and each router periodically runs a TE algorithm to compute path splits.
While it can react faster than centralized TE, dSDN requires \emph{every} switch to solve TE for \emph{all} demands in the network.
Routers have modest compute, so dSDN can only run TE heuristics~\cite{b4} and can thus be sub-optimal.
Its reaction and convergence time is constrained by how quickly demand information can disseminate through the network and how quickly the TE algorithm completes.

\subsection{\sysname{}}\label{sec:onlineTE}

%
In this paper, we explore a new capability, \textit{near-optimal online TE} (\sysname{}, for short).
\sysname{} decomposes the TE algorithm into two parts (\cref{fig:te_loops}), one which, motivated by the availability of switch compute~\cite{sigcomm2024:krentselskmnras24}, runs on switches, and another on a centralized controller (\cref{sec:design}).
Instead of running the TE algorithm periodically, it initiates the TE computation as soon as any demand on any switch changes.
This re-computation finishes \textit{in a few seconds} and is near optimal even on a large WAN (\cref{sec:evaluation}).
At the end of the re-computation, each switch knows the traffic splits for its demands, so it can install them locally.

\sysname{} is currently designed for TE algorithms used in ISP settings, MLU and Max-Flow.
We have left extending to other objectives (prioritization, fairness) to future work.

Our design addresses the three shortcomings discussed above.
Unlike existing centralized or distributed approaches, \sysname{} can quickly compute near-optimal TE solutions even for large WANs because it leverages massive parallelism by distributing the computation across switches (\cref{fig:te_loops}).
\sysname{} also reacts much faster than centralized TE for two reasons.
First, \sysname{}'s algorithms can \textit{warm-start} from the most recent solutions, avoiding the cold‑start overhead that slows down centralized TE methods.
Second, switches compute and install the demand splits locally, whereas centralized TE must program the switches across the entire network, incurring significant latency~\cite{sigcomm2024:krentselskmnras24}.

%

%


\sysname{} achieves this by building upon a long line of work on optimization decomposition~\cite{boyd11:_distr}.
We explain how these decompositions work with a simple example.

\parab{Illustrating optimization decomposition.}
\sysname{} decomposes a TE optimization into two parts: a {\it local} component that each switch can execute independently, and a {\it global} component that runs on the central coordinator (\cref{fig:te_loops}).
To illustrate this, consider the 4-node topology in \cref{fig:dd_example} with two demands: \textit{Demand 1} from $1 \rightarrow 4$ of 4 units, and \textit{Demand 2} from $2 \rightarrow 4$ of 2 units.
%
%
Suppose we wish to determine path splits for these demands with the MLU objective.


\begin{figure}[t]
\begin{tikzpicture}[
    >=Latex,
    node distance=2cm,
    every node/.style={font=\sffamily},
    v/.style={
        circle, minimum size=4mm,
        draw=blue!80, fill=none,
        line width=1pt, inner sep=0pt,
        text=black, font=\bfseries
    },
    ecap/.style={midway, inner sep=1pt, fill=white, font=\small},
    thinedge/.style={line width=0.8pt},
    thickflow/.style={line width=2.0pt}
]

\node[v] (v1) at (0,0) {1};
\node[v] (v2) at (0,1.3) {2};
\node[v] (v3) at (1.3,0) {3};
\node[v] (v4) at (1.3,1.3) {4};
\draw[-,thinedge] (v1) -- (v2) node[ecap,] {2};
\draw[-,thinedge] (v2) -- (v4) node[ecap] {2};
\draw[-,thinedge] (v1) -- (v4) node[ecap] {4};
\draw[-,thinedge] (v1) -- (v3) node[ecap] {2};
\draw[-,thinedge] (v3) -- (v4) node[ecap] {2};

\draw[->,thickflow] (0.2,0.5) -- (.9,1.2)
    node[ecap, above=0pt, sloped, font=\bfseries] {4};
\draw[->,thickflow] (0.3, 1.6) -- (1.1, 1.6)
    node[ecap, above left=2pt and -2pt, font=\bfseries] {2};

\node[anchor=west] at ($(current bounding box.east)+(-0.3cm,-0.15cm)$) {%
\begin{minipage}{7cm}
\centering
\scalebox{0.65}{
\begin{tabular}{|p{15mm}|c|c|c|c|c|}
\hline
\textbf{Round} & \textbf{0} & \textbf{1} & \textbf{2} & \textbf{.} & \textbf{N} \\
\hline\hline
\textbf{Demand1} & (0, 4, 0) & (0.1, 3.4, 0.5) & (0.2, 2.8, 1) &  & (0.25, 2.25, 1.5) \\
\hline
\textbf{Demand2} & (2, 0, 0) & (1.75, 0.25, 0) & (1.5, 0.5, 0) &  & (1.25, 0.75, 0) \\
\hline
\textbf{Price1} & 2 & 1.85 & 1.7 &  & 1.5 \\
\hline
\textbf{Price2} & 4 & 3.65 & 3.3 &  & 3 \\
\hline
\textbf{Price3} & 0 & 0.5 & 1 &  & 1.5 \\
\hline
\textbf{MLU} & 1 & 0.925 & 0.825 &  & 0.75 \\
\hline
\end{tabular}}
\end{minipage}
};

\end{tikzpicture}
\caption{Illustrating optimization decomposition. \textit{Demand 1} is routed along paths $1-2-4$, $1-4$ and $1-3-4$, while \textit{Demand 2} is routed along paths $2-4$, $2-1-4$, and $2-1-3-4$. The splits in the above table are for paths in this order. Prices correspond to links or path segments shared between the demands, respectively $2-4$, $1-4$, and $1-3-4$.
 \label{fig:dd_example}\label{tab:rounds_cold} 
}
\end{figure}
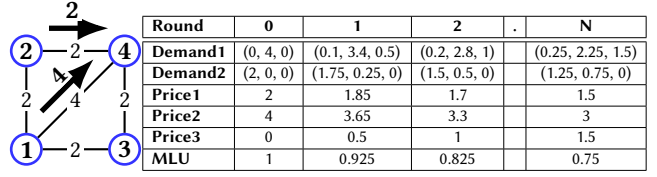
\begin{figure}[t]
  \scalebox{0.6}{
\begin{filecontents*}{csvfiles/dd_warm.csv}
;Initial;Update;Round 1;Round 2;...;Round N
\textbf{Demand 1};(0.25, 2.25, 1.5);(0.25, 1.75, 1.5);(0.15, 1.85, 1.5);(0.05, 1.95,1.5);;(0, 2, 1.5)
\textbf{Demand 2};(1.25, 0.75, 0);(1.75, 0.75, 0);(1.65, 0.85, 0);(1.55, 0.95, 0);;(1.5, 1, 0)
\textbf{Price 1};1.5;2;1.8;1.6;;1.5
\textbf{Price 2};3;2.5;2.7;2.9;;3
\textbf{Price 3};1.5;1.5;1.5;1.5;;1.5
\textbf{MLU};0.75;1;0.875;0.8;;0.75
\end{filecontents*}
\csvreader[%
tabular = |l|c|c|c|c|c|c|,%
table head = \hline & \textbf{Initial} & \textbf{Update} & \textbf{Round 1} & \textbf{Round 2} & \textbf{...} & \textbf{Round N} \\ \hline\hline,%
separator = semicolon,%
late after line = \\ \hline%
]{csvfiles/dd_warm.csv}{}{\csvlinetotablerow}}
  \caption{When Demand 1 changes to 3.5 and Demand 2 to 2.5, the algorithm can \textit{warm-start} from the previous solution.} 
  \label{tab:rounds_warm}
\end{figure}

To solve this optimization in a distributed fashion, consider the following iterative approach, which proceeds in rounds.
In the first round, each switch initializes the splits, then sends these to a central coordinator.
The coordinator sets \textit{prices} on path segments on which the demands overlap, where the price reflects the relative utilization on the bottleneck link on the path segment.
In subsequent rounds, each switch updates the splits to lower the prices, then sends these to the coordinator, and the process repeats.
Eventually, the switches find splits that converge to the optimal MLU (0.75 in our example).

\cref{tab:rounds_cold} shows the splits (per path listed above) for each demand, the prices for the path segments, and the evolution of the MLU at each round.
Now suppose Demand 1 changes to 3.5 and Demand 2 to 2.5.
This approach can converge quickly since the computation can \textit{warm-start} from the previous solution (\cref{tab:rounds_warm}).




\parab{Challenges and Opportunities.}
The theoretical foundations for decomposing optimizations were developed decades ago~\cite{boyd11:_distr,bertsekas2015parallel}.
However, we know of no prior work that has exploited optimization decomposition to devise a completely \textit{distributed} algorithm for TE.

\cref{fig:dd_example} uses the method of \textit{dual decomposition}~\cite{Rush_2012} because the coordination between controller and switches via prices is easier to understand.
\sysname{}, however, uses a variant, the Alternating Direction Method of Multipliers (ADMM)~\cite{boyd11:_distr}, to decompose an optimization similar to \cref{fig:dd_example}.
ADMM provably converges quickly to a near-optimal solution.
Even so, adapting ADMM to the networking setting presents four major challenges (\cref{sec:design}).
%


First, TE algorithms often use different objectives (Max-Flow, MLU), and \sysname{} must support these.
In each case, determining the decomposition (\textit{i.e.}, which part of the algorithm runs on the switches, and which on the controller, as in \cref{fig:dd_example}) can be tricky and involve complex mathematics.

%
Second, the decomposition must scale well in the volume of information exchanged between the controller and the switches.
A generalized version of our example in \cref{fig:dd_example} requires a price per edge per demand.
In a large ISP network with hundreds of links and tens of thousands of demands, this can be prohibitively expensive.
The decompositions must also be computationally-efficient.
In general, in optimization decomposition, each component also runs an optimization.
The component that runs on the switches must be lightweight enough to run on switch CPUs.

Third, in a large network, switches and the controller might need many more iterations to converge than in our example (\cref{fig:dd_example}), which can result in large convergence times especially for continental-scale WANs.
\sysname{} must converge quickly.
In particular, the algorithm described in \cref{tab:rounds_cold} is \textit{synchronous}; the controller must wait for each switch to finish a round before proceeding to the next.
Large delays between a controller and a \textit{single} switch in a large WAN can adversely impact convergence.

\begin{table}[t]
    \centering
    \scalebox{0.65}{
    \begin{tabular}{|l|c|c|c|} \hline
      & Optimality-Guarantee & Highly-reactive & Edge-Based TE  \\ \hline\hline
     Optimal TE & \checkmark & $\times$ & $\times$ \\
     Heuristic~\cite{dempin,blastshield,ncflow,pop,soroush} & $\times$ & $\times$  & $\times$ \\
     Approximation~\cite{soroush,swan} & \checkmark & $\times$ & $\times$ \\ 
     Learning-based~\cite{teal_2023,perry23:_dote,neuralwan} & $\times$ & $\times$  & $\times$ \\
     DeDe~\cite{xu2025decoupledecomposescalingresource} & \checkmark & $\times$ & $\times$ \\ 
     COpter~\cite{coopter} & \checkmark & $\times$ & $\times$ \\ 
     dSDN~\cite{sigcomm2024:krentselskmnras24} & $\times$ & \checkmark & $\times$ \\
     TeXCP~\cite{texcp}, MATE~\cite{mate} & $\times$ & \checkmark & $\times$ \\ \hline \hline
     \textbf{\sysname{}} (ours) & \checkmark & \checkmark  & \checkmark \\ \hline
    \end{tabular}
    }
    \caption{Of the most closely related work, \sysname{} is the only near-optimal method, capable of reacting to changes online, and of supporting edge-based formulations.}
    \label{tab:comparison}
\end{table}

Finally, the optimization decomposition presents us with an interesting opportunity.
Today, most large-scale TE is \textit{path-based}, and uses a small number of pre-determined paths between each ingress-egress pair (usually 4 to 16~\cite{krishnaswamy23:_onewan,denis23:_ebb,soroush}).
The TE algorithm then determines optimal splits for these paths.
An \textit{edge-based} approach jointly determines paths as well as splits.
Because it is not restricted to a few paths, it can achieve better values of the objective (\textit{e.g.}, lower MLU).
Edge-based approaches are more computationally intensive, so they have not been considered practical for TE.
However, the use of a distributed optimization algorithm like ADMM effectively \textit{parallelizes} the computation, rendering edge-based TE feasible.
The key challenge is to ensure that the paths chosen by the edge-based TE are low-latency paths.

\parab{Contributions.}
We make the following contributions:
\begin{itemize}
\item We present the first known \textit{distributed solvers} for path-based (\cref{sec:basic-distributed-te}) and edge-based TE formulations (\cref{sec:edge-based-solvers}), for MLU and Max-Flow.
These solvers are both computation and communication-efficient; in particular, the switch-side computations require significantly less compute resources than that present in modern routers.
\item Our suite of solvers, \sysname{}, employs a unique hierarchical ADMM decomposition (\cref{sec:ensur-fast-conv}) common to both path-based and edge-based formulations that allows it to scale, while ensuring fast convergence despite large WAN latencies.
\item Our evaluation (\cref{sec:evaluation}) using a complete implementation\footnote{We plan to make the implementation publicly available.} of \sysname{}, on a testbed emulating the KDL WAN, the largest topology in the Internet Topology Zoo~\cite{itzoo}, demonstrates an order of magnitude performance improvement in some cases over centralized TE solvers.
\end{itemize}

\cref{tab:comparison} highlights how \sysname{} is unique among related work.
No prior work has demonstrated the ability to solve edge-based TE at scale.
Relative to fast centralized TE solvers\cite{ncflow,pop} or decentralized ones~\cite{sigcomm2024:krentselskmnras24,mate,texcp}, \sysname{} is near-optimal and can react much faster.
ML-based TE solvers~\cite{teal_2023,perry23:_dote} can be fast, but sub-optimal~\cite{learning-enabled}.
DeDe~\cite{xu2025decoupledecomposescalingresource} that employs an ADMM decomposition to demonstrate a fast, parallelizable solver for centralized path-based TE.
We discuss related work in more detail in \cref{sec:related-work}.

\parab{Ethics.} This work does not raise any ethical issues.


\section{\sysname{} Design}\label{sec:design}


In this section, we first introduce ADMM and then provide an overview of \sysname{}.
Subsequent sections describe our solvers' design.
Our descriptions of ADMM and our designs are informal; appendices contain mathematical details.


\subsection{Introduction to ADMM}\label{sec:introduction-admm}

ADMM~\cite{boyd11:_distr} solves an optimization problem by \textit{decomposing} it into sub-problems, then iteratively solving each sub-problem in turn.
Depending on the optimization problem and its decomposition, it may be possible to solve some of the sub-problems in a distributed manner.

\begin{figure}[t]
  \centering
  \includegraphics[width=1.0\columnwidth]{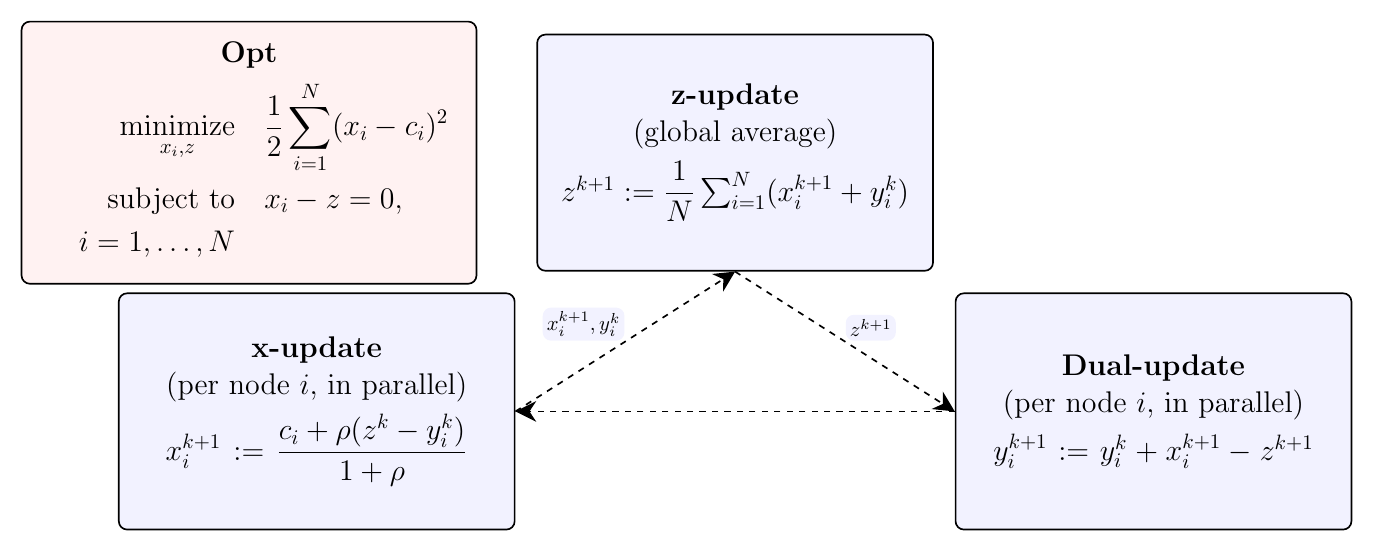}
  \caption{\label{fig:admm_consensus} ADMM decomposition for a simple Consensus problem.}
\end{figure}

\cref{fig:admm_consensus} illustrates an optimization problem for \textit{distributed averaging}, \textit{i.e.}, finding the mean of a set of values distributed across the network.
It is somewhat contrived---there are better ways of averaging---but  illustrates ADMM decomposition.
Consider a network with $N$ nodes, where each node $i$ holds a local value $c_i$ and a variable $x_i$.
Our goal is for every $x_i$ to converge to the arithmetic mean of the $c_i$s, enabling all nodes to compute the global mean in a distributed fashion.

In \cref{fig:admm_consensus}, \textbf{Opt}'s objective exploits the property that the sum of the squared deviations of any number $z$ from a set of $N$ numbers (\ie $c_i$s) is smallest when $z$ is the mean of the $N$ numbers.
The constraint ensures that all $x_i$s eventually converge to the mean $z$.

%

The fundamental problem is that no node in the system knows all the values of $c_i$ and of $z$.
ADMM allows nodes to iteratively converge on the global mean.
It does so by decomposing the problem into three steps  (\cref{fig:admm_consensus}).
These three steps execute iteratively one after another.
At the end of each iteration, intuitively, all nodes move a step closer to $z$, at a rate determined by a step size $\rho$.

In the $k+1$-th iteration, the \textbf{\xupd{}} step, moves $x_i$ towards the current estimate of the global mean $z^k$.
It does this using a \textit{dual variable} $y_i$ for each node, which captures, at any given instant, how far off a node's estimate $x_i$ is from $z^k$.
Then, each node conveys the updated value of $x_i$ to a \textit{central coordinator}, which updates the current estimate of $z^k$ (each node also sends its $y_i$ estimate to the coordinator), a step called the \textbf{\zupd{}}.
Subsequently, the coordinator sends the new estimate of $z^{k+1}$ back to all the nodes.
Each node first updates $y_i$ (the \textbf{\dupd{}}), then repeats the cycle again.
It is easy to see this \textit{ADMM loop} would eventually result in $x_i$s converging to the mean of $c_i$s.
%


This example highlights several properties of ADMM that we exploit for \sysname{}.

\begin{itemize}[leftmargin=*,noitemsep]
    \item Given an optimization formulation, an ADMM decomposition consists of three update steps, executed iteratively as described above~\cite{boyd11:_distr}.
    With rather mild assumptions, the iterations provably converge to an optimal solution; more important, for an appropriate choice of step size $\rho$, the computation gets very close to the optimal in a few steps~\cite{boyd11:_distr}.
    
    \item We can \textit{warm-start} the ADMM loop in \cref{fig:admm_consensus}. For example, if $c_i$s at some nodes change, the technique can start from the previous values of $x_i$ to converge to the new optimum.
    
    \item The decomposition in \cref{fig:admm_consensus} is amenable to a distributed implementation: all the nodes can perform the \xupd{} and \dupd{} steps \textit{in parallel}, while the \zupd{} is centralized.
    This requires modest \textit{communication}: at each iteration, each node sends two values ($x^{k+1}_i, y^{k}_i$) to the coordinator and receives a single value $z^{k+1}$.
    
    \item Each update has lower \textit{computational complexity} than \textbf{Opt}, which is a quadratic optimization.
    In this example, all the steps have closed-form solutions.
    In other cases, the \xupd{} and \zupd{} may require solving an optimization.
\end{itemize}

\vspace*{-0.5ex}
\subsection{\sysname{} Overview}\label{sec:sysname-overview}


\sysname{} provides a family of distributed TE solvers, each of which is both near‑optimal and online.
%
%
It supports both path‑based~\cite{redTE,teal_2023} and edge‑based TE formulations.
%
%
\sysname{} can also optimize two objectives widely used in TE: MLU and Max-flow (\cref{s:introduction}).



\parab{Requirements.}
In addition to being near-optimal, \sysname{}'s decompositions must enable solvers to be \textit{online}---when a demand change or link failure happens, the solver must (a) start re-optimizing the demand splits immediately, instead of periodically (\cref{s:introduction}), (b) the re-optimization must converge quickly and (c) when the re-optimization completes each switch must be able to \textit{effect} the revised demand splits by locally programming its flow tables; in other words, the decomposition should not require coordination between switches to install flow state.
All of these enable fast reactivity to changes.
In general, (b) may not always be possible with ADMM; in each iteration, the \xupd{} and the \zupd{} may require solving an optimization.
Heavyweight optimization solvers (\textit{e.g.}, those that use interior point methods for linear programs~\cite{gondzio_ipm_warm}) re-solve the optimization from scratch at each iteration, and can delay convergence.
\sysname{} must find more computationally-efficient decompositions, especially for the switch-side optimization, and leverage \textit{warm-starting} from previous solutions to speed up convergence.





\sysname{} should \textit{scale} to large networks~\cite{krishnaswamy23:_onewan}. 
There are two components to scaling:
\textit{Communication efficiency} demands that the communication cost of the distributed computation scales well with network size and the number of demands.
\textit{Computation efficiency} requires that the per-switch computations remain feasible given the computing capabilities of today's network devices.

Our solution must \textit{converge quickly even as the network scales}.
Each iteration can finish only after the coordinator has received updates from all switches, so overall convergence is constrained by the maximum latency between the coordinator and any switch.
As network size increases, so can this latency, and \sysname{} must incorporate techniques to converge quickly despite this increased latency.

Finally, we must preserve \textit{uniformity} to simplify implementation.
This means the ADMM components must be re-usable: where possible, \xupd{} or \zupd{} steps must be near-identical, or use the same type of solver.
Control structures must follow similar patterns, such as the coordinator-orchestrated pattern of~\cref{fig:admm_consensus}.

\section{The Edge-Based Solver for MLU}\label{sec:edge-based-solvers}

We first describe \sysname{}'s online edge-based solver for MLU, because this is the easiest to explain intuitively.
\cref{app:edge-based-form} contains the detailed mathematical descriptions. 

\parab{Edge-Based TE}
takes a set of demands as input, each representing traffic between a pair of switches.
It attempts to route these demands while \textit{minimizing} MLU.
Path-based TE assumes a pre-configured set of paths between every switch pair but edge-based TE jointly optimizes the paths and the allocations for every demand.
As a result, the edge-based problem is harder to scale, which has prevented its adoption in deployed systems.
If a scalable solution exists, it would strictly outperform (\textit{e.g.}, have lower MLU or higher total flow) path-based methods.

A solution to the edge-based TE problem that optimizes MLU must satisfy two sets of constraints.
%
%
Demand constraints ensure that every demand is \textit{fully routed}: the solver assigns sufficient paths and capacities for each demand.
Flow conservation constraints ensure that, for every demand, the total ingress traffic at every intermediate switch equals the total egress traffic: traffic does not disappear in the network.

\parab{A Strawman Distributed Edge-based TE.}
With an appropriate ADMM decomposition, edge‑based TE becomes tractable because the optimization can be parallelized across hundreds or thousands of switches.
However, identifying a decomposition that is both correct and amenable to distributed implementation is challenging.

A strawman approach is to apply an ADMM decomposition similar to that of \cref{fig:admm_consensus}, but with different variables and update steps.
The \zvar{} variable is a matrix representing the total flow per demand on each edge. 
The $y$ variable encodes the price of sending an additional unit of flow through an edge.
The \xvar{} variable is also a matrix, where the $(i,j)$-th entry denotes the amount of demand $i$ routed on edge $j$.

In each iteration, on the controller, the \zupd{} uses the current \xvar{} values to update \zvar{} so as to minimize MLU, and the $y$-update sets prices on each edge.
On each switch, the \xupd{} adjusts per-edge allocations to minimize the total price.
Specifically, each switch determines how to split each demand across its egress ports.
This moves demands away from heavily congested edges.
The magnitude of each update is controlled by the step size $\rho$ (as in \cref{fig:admm_consensus}).

The prices enable ADMM to alternately refine \xvar{} and \zvar{}, allowing each to be updated using a separate optimization.

Unfortunately, this strawman approach \textit{does not scale well}.
\zvar{} is a matrix whose dimension grows with the product of number of demands and edges, and \xvar{}$_i$ on each switch grows with the number of demands.\footnote{Each switch computes and sends only the columns of the \xvar{} matrix that correspond to its egress edges.} 
A large WANs with a thousand nodes~\cite{krishnaswamy23:_onewan} can have up to a million demands, so \zvar{} and \xvar{}$_i$ can be matrices with up to a billion and a million entries respectively.
Each ADMM iteration involves sending \zvar{} from the controller to all the switches and \xvar{}$_i$ from all the switches to the controller.
So, transmitting such large matrices significantly slows down convergence.
Furthermore, since transmitting \xvar{} and \zvar{} represents control traffic, this approach would incur significant control overhead.



\parab{Towards Practical Edge-based TE.}
We improve the communication-efficiency of the edge-based solver as follows.
Observe that a switch does not need to know the per-demand allocation on each edge for \textit{other switches}; it only needs to know the total flow allocated (or the utilization) on each link.
Accordingly, in \sysname{}, the \zvar{} encodes the total flow allocated to each edge.
As a result, the controller sends to every switch a vector whose size scales only with the number of edges rather than with the product of edges and demands.

To reduce the size of \xvar{}, instead of requiring each switch to compute split ratios for all demands, \sysname{} asks each switch to \textit{compute routing decisions only for the demands that originate at that switch}, but across all links.
This design has three advantages.
First, it reduces the size of \xvar{}$_i$ from the number of global demands to the number of edges.
Second, because allocations can be biased to lie on near‑shortest paths for low latency (see below), the resulting \xvar{} is sparse, making the actual communication cost far smaller.
Finally, at the end of the optimization, each switch (a) can program the demand splits locally (\textit{i.e.}, in a consensus-free manner~\cite{sigcomm2024:krentselskmnras24}) and (b) knows the paths allocated to the splits and can use source-routing to send packets~\cite{sigcomm2024:krentselskmnras24}.
This renders the edge-based solver practicable in WAN settings.

\begin{figure*}[t]
  \begin{minipage}[t]{0.33\textwidth}
    \centering
    \includegraphics[width=0.85\columnwidth]{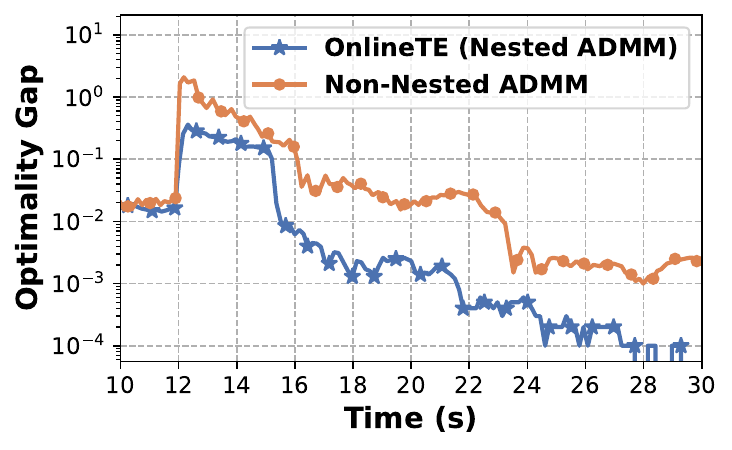}
    \caption{\sysname{}'s nested ADMM converges faster than a single ADMM loop. \label{fig:non-nested}}
  \end{minipage}
  \begin{minipage}[t]{0.33\textwidth}
    \centering
    \includegraphics[width=0.95\columnwidth]{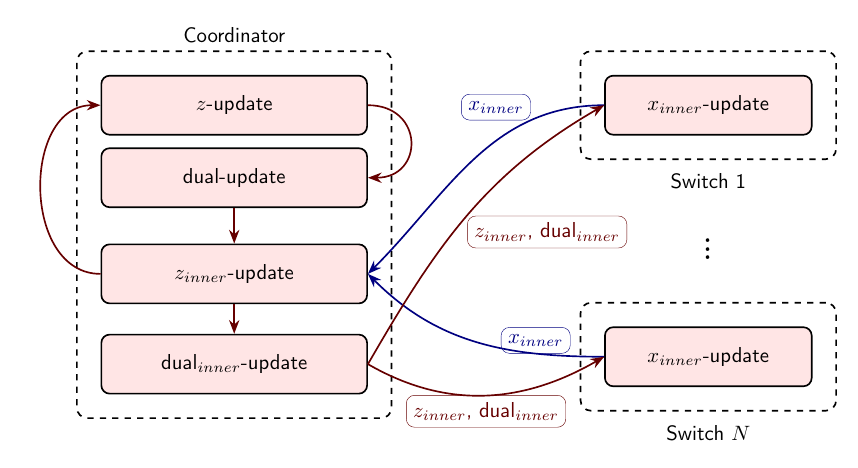}
    \caption{\label{fig:nested_admm} Nested ADMM for Edge-based MLU.}
  \end{minipage}
  \begin{minipage}[t]{0.33\textwidth}
    \centering
    \includegraphics[width=0.8\columnwidth]{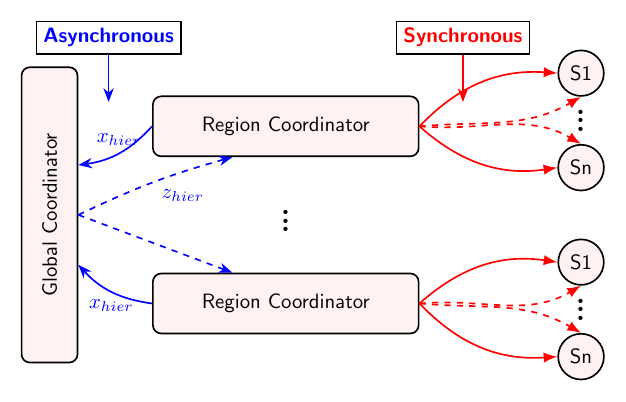}
    \caption{\sysname{}'s coordinator hierarchy with inter-region asynchronous ADMM. \label{fig:hierarchy} }
  \end{minipage}
\end{figure*} 

\parab{Nested ADMM: Improving convergence further.}
We have found that this decomposition converges slowly, because the switches are \textit{chasing a moving target}.
In one iteration, \zupd{} sets target utilizations on links that optimize MLU, and the \xupd{} moves switch allocations towards these prices.
However, in the next iteration, \zupd{} updates these target utilizations, and the switches must now readjust their allocations to move towards these new targets.



We have found that faster convergence can be achieved by introducing a \textit{nested ADMM} loop.
In the inner loop, switches take multiple steps to move towards a given, fixed, target.
Then, the outer-loop iteration runs once to re-adjust the target, and the inner-loop repeats again.
\cref{fig:non-nested} illustrates the benefit of this approach: the optimality gap (difference between the optimal and the distributed solver's MLU) decreases much faster over time (y-axis in log-scale) for \sysname{}'s nested approach than a non-nested approach.

In the \textit{inner} ADMM loop, the \xvar{} variable, denoted \xivar{}, is the same as the outer ADMM loop's \xvar{} variable which captures the assignment of demands to network edges.
However, \zivar{} is a variable that tracks the target mean of all commodities on an edge, and the \divar{} variable tracks the difference between this mean and each switch's assignment.
This effectively ensures that the \xiupd{} is decoupled from the \zupd{}.
Specifically, the inner ADMM loop converges to a target mean per edge (using $\rho_{inner}$ as the step size), the outer ADMM loop revises that target to move allocations closer to the optimal (using $\rho_{outer}$ as the step size), the inner ADMM loop revises allocations to meet the new targets, and so on.
Unlike the Consensus problem of \cref{fig:admm_consensus}, the inner ADMM loop solves a \textit{Sharing} problem~\cite{boyd11:_distr} where switches collaborate to share allocations on each edge.


This nested ADMM loop structure does not increase communication complexity relative to the non-nested decomposition.
Each switch send its \xivar{} to the coordinator; these contain edges on which the switch's demands are non zero.
The coordinator sends \zivar{} and \divar{} to the switches, each a vector with one entry per edge.

\cref{fig:nested_admm} depicts where each part of the inner and outer ADMM decompositions is executed, and what information the switches and the coordinator exchange.
Only the \xiupd{} runs on the switches; all other components run on the coordinator.
As discussed above, to achieve fast convergence, we need to iterate on the inner ADMM loop a few times before running one instance of the outer ADMM loop.


\parab{Improving Computation Efficiency.}
Unfortunately, in the inner ADMM decomposition, the \xiupd{} optimization is too heavy-weight.
It contains numerous variables and constraints to ensure flow conservation and demand satisfaction.
To reduce its computational complexity, we obtain a null space basis of the topology adjacency matrix, which requires a one-time singular value decomposition (SVD).
We then express the demand satisfaction constraint in terms of the basis vectors of this matrix.
This technique reduces the number of variables to $O(|E|-|V|)$, where $|E|$ is the number of edges and $|V|$ is the number of nodes.
When formulated this way, flow conservation does not require separate constraints.
\cref{app:edge-based-form} describes the mathematical details.

%
Without this optimization, solving \xiupd{} may require complex interior point methods~\cite{boyd_co} using a commercial solver~\cite{gurobi,applegate2022PDLP}.
Such methods cannot be easily warm-started, so our solver cannot run online.
With this optimization, we can use simpler methods described below.


\parab{Computational Complexity.}
This final decomposition (using nested ADMM and the null space decomposition) also has attractive computational complexity properties.
The most heavyweight component is the \zupd{}.
This is an inequality constrained, convex quadratic program (QP) with the number of variables equal to the number of edges.
For \zupd{}, we use PDLP~\cite{applegate2022PDLP}, a recent open-source solver designed for large-scale optimization problems with good warm-start performance.
In \sysname{}, the coordinator can be implemented in an SDN-controller, whose computational resources are sufficient to run PDLP.

The \xiupd{} is a very lightly constrained convex QP with $O(|E|-|V|)$ variables.
It can be parallelized across demands at a single node.
Each instance can be solved using simple algorithms like projected gradient descent~\cite{boyd_co}, which can easily be warm-started.
As we show in \cref{sec:evaluation}, these can run on moderately powerful switch CPUs~\cite{sigcomm2024:krentselskmnras24}.
The \ziupd{} has a computationally lightweight closed-form solution, and the dual variable updates for both ADMM loops involve simple vector operations.

\parab{Controlling path stretch.}
Edge-based TE can better optimize the global objective than path-based TE, but may allocate long paths to some demands, increasing latency.
To address this, our edge-based solver adds an $\ell_1$-norm regularizer to the \xiupd{}.
This has the effect of incentivizing the solver to return \textit{sparse} allocations (those with mostly zero entries).
A parameter $\kappa$ controls the degree of sparsity, and indirectly, the path stretch, since if demands are allocated to fewer edges, they are less likely to encounter long paths.
This regularizer applies both to edge-based Max-flow (\cref{sec:basic-distributed-te}) and MLU.
$\kappa$ is a function of the topology, not the objective.

The regularizer increases the computational complexity of the \xiupd{} slightly, so that, instead of projected gradient descent, it requires a more complex coordinate descent~\cite{wright2015coordinatedescentalgorithms} solver.
This can run on switch CPUs and can warm-start.
Regularization does not change communication complexity.




\section{Extensions: Max-flow and Path-based TE}\label{sec:basic-distributed-te}

\parab{The Edge-based Max-flow solver.}
Edge-based Max-flow requires relatively small changes to the nested ADMM structure of \cref{fig:nested_admm}.
Specifically, it introduces a new variable $\beta$ to the \xiupd{}, a vector of dimensionality equal to the total number of demands.
$\beta_k$ specifies what fraction of the demand is currently being routed.
This is paired with an $\alpha$ vector in the \zupd{} which tracks target fractions.
These terms appear in the dual variables as well.

These additions do not significantly increase the communication cost.
After the \xiupd{}, the switch must send the $\beta$ values for its own demands to the controller.
The number of demands at a switch is proportional to the number of nodes in the network.
The computational complexity of all components remains the same as for MLU, except for the \xiupd{}.
This is a QP with a coupling term (\xivar{} and $\beta$ are coupled), for which we use coordinate descent~\cite{wright2015coordinatedescentalgorithms} which can also be warm-started.


\parab{Path-based TE.} 
%
%
\sysname{} includes MLU and Max-flow solvers for path-based TE.
\cref{app:path-based-form} describes the mathematical details underlying these solvers.

Both use the nested decomposition described in \cref{fig:nested_admm}.
More important, the \zupd{} in the outer ADMM loop is identical to that in edge-based TE.
This is possible because a path-based assignment can be converted into an edge-based assignment.
There are two benefits to this.
It allows us to re-use components from edge-based solvers thereby addressing our uniformity requirement.
Moreover, it preserves the communication complexity of edge-based solvers which scales with the number of edges.

The inner ADMM loop changes slightly.
Specifically, \xivar{} now encodes demand splits per path instead of per edge, and \zivar{} encodes the capacity headroom available per path.
The dual variables do not change.
Finally, the complexity of \xiupd{} changes slightly.
Recall that this step uses projected gradient descent.
%
%
In the path-based solvers, the projection step is more expensive because the \xiupd{} problem is more constrained; the projection for edge-based TE in contrast is very trivial.
Despite this, the path-based solver is still faster, as it handles much fewer variables.

Finally, the inner ADMM converts the path-based assignments to an edge-based assignment.
To do this, in practice, each switch could monitor the edges on its pre-configured paths using MPLS traceroute.
In settings where tunnels are configured by an SDN controller~\cite{denis23:_ebb,krishnaswamy23:_onewan}, the switch agent would have this information.

\section{Ensuring Fast Convergence at Scale}\label{sec:ensur-fast-conv}

As described so far, the inner ADMM loop is \textit{synchronous}; the coordinator must wait for all switches to respond before processing a \zupd{}.
As the network scales, the time to complete one iteration is a function of the maximum latency between a switch and a coordinator.
A \textit{straggler} can significantly slow convergence, a problem faced by many synchronous distributed computations (\textit{e.g.}, LLM training).

\parab{Asynchronous ADMM.}
To address this drawback, prior work~\cite{zhang2014asynchronous,Wei2013ADMMConvergence} has considered methods that introduce asynchrony in ADMM loops.
In these methods, the coordinator runs a \zupd{} even without receiving updates from all switches.
For \sysname{}, we have adapted the \textit{partial-barrier} asynchronous ADMM technique designed for consensus problems~\cite{zhang2014asynchronous} to our inner ADMM loop which solves a sharing problem.
In partial-barrier, the coordinator can apply a \zupd{} after a \textit{minimum} of $k$ nodes respond.
It then sends the updated \zvar{} only to the responders.
In addition, any given node must respond at least once within $\tau$ iterations.
Partial-barrier still ensures convergence to the optimal~\cite{Wei2013ADMMConvergence}, and convergence time is proportional to $\tau$ and inversely proportional to $k$.
It can be used in all the solvers discussed above, edge-based or path-based, for MLU and Max-flow.

Unfortunately, even this asynchronous ADMM is insufficient to ensure fast convergence.
On a large WAN with over 700 nodes (\cref{sec:evaluation}), we were unable to observe convergence with the partial-barrier method unless $k$ was very large.

\parab{Exploiting Network Structure with Hierarchical ADMM loops.}
%
%
Due to geographical constraints, switches naturally cluster into \textit{regions}, with the property that switches within a region (\textit{e.g.}, a metro area or a state) are closer to each other than to switches in other regions. 

To speed up convergence, \sysname{} structures large-scale edge-based or path-based TE problem as a \textit{hierarchy of ADMM loops} (\cref{fig:hierarchy}).
Within a region, a region coordinator and switches within the region run the nested ADMM loop of \cref{fig:nested_admm}.
A second \textit{inter-region ADMM loop} between a global coordinator and the region controllers coordinates inter-region allocations. 
All intra-region ADMM loops are synchronous, and the inter-region ADMM loop is asynchronous.
If one, or a few regions are stragglers, asynchronous ADMM will still permit fast convergence on allocations that impact only the non-straggler regions.
For example, in path-based TE, if a path traverses regions A, B, and C, a straggler region D cannot impact the convergence of allocations on that path.


This inter-region ADMM loop solves a consensus using a structure identical to \cref{fig:admm_consensus}.
Each region coordinator (\cref{fig:hierarchy}) runs an \xhupd{}, where the \xhvar{} is a vector which contains the total flow allocated for traffic originating within this region on all edges in the network.
The global coordinator runs the \zhupd{}, in which the \zhvar{} is an $R \times |E|$ matrix, where $R$ is the number of regions.
The $(i,j)$-th entry in this matrix is the amount of flow on edge $j$ that region $i$ may \textit{not} allocate, \textit{i.e.}, it is reserved for other regions.
Intuitively, the \xhupd{} step allows regions to use available capacity without violating capacity constraints.
A region does not need to know how many demands each other region has, and how these demands are allocated.
The \zhupd{} step sets reserved capacities to bias allocations towards the global objective (MLU or Max-flow).
The dual variable incentivizes \textit{lagging} regions---those whose allocations have not fully utilized available capacity---to catch up.

The inter-region loop's \xhvar{} and \zhvar{} are the same for all solvers.
Different solvers would update these variables differently, depending on the problem or objective.
The communication complexity of the inter-region ADMM loop is proportional to $R \times |E|$; each region coordinator sends and receives a vector of length $|E|$.
In a network with a thousand edges, this information is about 4-8 KB.
The computational complexity of the \xhupd{} and the \zhupd{} is also minuscule.
In fact, these complete quickly enough that we do not need a solver with warm-start capabilities.
Moreover, these updates run on coordinators that can be co-located with SDN-controllers with sufficient compute capacity.

\section{Failure Handling}\label{sec:failure-handling}

\parab{Switch or Link failures.}
In addition to demand changes, \sysname{} can also react immediately to link and switch failures and converge to near-optimal allocations.
Consider a link failure.
In path-based TE, if a switch has a path traversing the failed link, it can detect the failure using BFD~\cite{bfd_rfc}, set the demand allocation on that path to be zero, redistribute that allocation locally, then trigger the ADMM loops to ensure convergence.
This approach would also apply to TE systems that pre-configure backup paths~\cite{denis23:_ebb,krishnaswamy23:_onewan}; prior to the failure, \sysname{} would not allocate demand on these paths, but would do so after the failure.
In edge-based TE, a link failure changes the topology adjacency matrix, and the switch re-factorizes the previous allocations to the new adjacencies.
This re-factorization percolates to other switches (and other regions) during the ADMM iterations, resulting in a near-optimal allocation.
A switch failure can be treated as the simultaneous failure of its attached links, and the steps described above would apply.

\parab{Coordinator failure.}
When a coordinator fails, a backup coordinator can resume from a \textit{snapshot} of the ADMM computation.
After every iteration, region coordinators can snapshot the most recent \xivar{}s from all the switches, the \zivar{} and the outer-loop \zvar{}, and its own inter-region \xhvar{} as well as all dual variables.
The global coordinator can snapshot, after every iteration, the most recent \xhvar{}'s from the regions, the \zhvar{}, and the dual variable.
Switches do not need to snapshot state; we assume the switch's ADMM computation shares fate with the switch.

\parab{Packet losses.}
Switch-to-coordinator communication can use RPC, as can coordinator-to-switch communication in the inter-region loop.
In regions with many switches, using RPCs between a coordinator and each switch can introduce overhead, so \sysname{}'s current implementation uses multicast.
This is possible because \zivar{} is small (about 4-8KB in a network with a thousand links).
In this case, to counter packet losses, the coordinator repeats the multicast transmission a few times, then retransmits if some switches do not respond within a fixed timeout.

\section{Evaluation}\label{sec:evaluation}

In this section, we compare \sysname{} against several baselines, demonstrate its ability to solve edge-based TE at scale, and quantify the importance of its hierarchical design.


\subsection{Methodology}\label{sec:methodology}

\parab{Implementation.}
We have implemented all of \sysname{}'s solvers in 22K lines of Python code.
Our implementation uses common BLAS libraries (available via \texttt{Numpy}) for switch and coordinator computations, except for the path-based solver which uses JIT compilation via \texttt{Numba} to optimize a sparse matrix multiplication.

\parab{Testbed.}
In all of our evaluations, we run \sysname{} on the SPHERE testbed~\cite{sphere}.
On this testbed, we materialize switches and controllers for two WAN topologies (described below) as separate virtual machines running a light-weight Debian Bullseye operating system.
Each VM runs on an AMD EPYC 7702 2.0 GHz processor.
We provision each switch to have only 2 cores and 2 GB of memory; as such, they have less compute than modern WAN switches (Arista 7816 switches similar to the ones used in~\cite{sigcomm2024:krentselskmnras24} currently have 64 GB of DRAM and 8-core processor at 2.0 GHz~\cite{arista-datasheet}).

\parab{Topologies.}
We use two publicly available topologies from the Internet Topology Zoo~\cite{itzoo}.
Most of our experiments use the largest topology in this dataset, KDL, which has 754 nodes and 1790 edges and is spread across the east coast of the US (\cref{fig:topologies}).
We also evaluate the importance of hierarchy using an inter-continental topology, Cogentco, with 190 nodes and 486 links (\cref{fig:topologies}).
On the testbed, we emulate these topologies by ensuring that the latency on every link in the topology matches the propagation latency between the routers\footnote{More precisely, we compute speed-of-light propagation delay between router geo-locations.} at the ends of the link.
\textbf{\textit{In all of our experiments, we run \sysname{} on these emulated topologies on the testbed.}}

\parab{Baselines.}
We compare against 3 baselines: POP~\cite{pop} with 16 sub problems, DeDe~\cite{xu2025decoupledecomposescalingresource}, and NCFlow~\cite{ncflow}.
These represent different methods to speed up centralized TE computations.
Our \textbf{\textit{comparisons against the baselines is idealized}}: to match current practice, we invoke these baselines every 5 minutes, but we do not account for demand collection and switch programming latencies (\cref{fig:te_loops}).
We compare these baselines against \sysname{}'s path-based solver, since they cannot solve edge-based TE.
All path-based experiments use 16 pre-configured paths.
For other baselines listed in \cref{tab:comparison}, we could not find the code or, for the learning-based ones, models trained on our topologies and traffic matrices.

\parab{Traffic matrices.}
Since no publicly available traffic matrices exist for KDL and Cogentco, we
follow prior work~\cite{ncflow,soroush} and synthesize traffic matrices using Uniform, Gravity and Bimodal distributions.
We evaluate three network load levels by scaling each traffic matrix so that the optimal path-based MLU is 0.1 at Low load, 0.8 at Medium load, and 1.1 at High load.
Thus, at Low load, most approaches should be able to satisfy demands.
At Medium load, some heuristics may fail to satisfy demands, and, at High load, even the optimal solver cannot find a demand-satisfying assignment.

\parab{Emulating traffic dynamics and failures.}
We begin each experiment with a traffic matrix, and run \sysname{} for 10 minutes.
Some experiments evaluate traffic dynamics, and others evaluate link failures.
To model aggressive traffic variation, we select 5\% of the demands every 20 seconds, and perturb (increase or decrease) them randomly between the minimum and maximum entries of the initial traffic matrix.\footnote{We do not emulate actual traffic on the testbed, since it does not have enough capacity to carry traffic.}
This approximately preserves the Low, Medium, High distinctions.
%
%
To approximately preserve locality in the Gravity model, we sample another Gravity matrix every 20 seconds and copy 5\% of its entries (randomly chosen) into the current traffic matrix.
Over a 5-minute interval, nearly half the entries change by at least 50\% on average, comparable to dynamics observed in large WANs (nearly 40\% within a 5-minute interval~\cite{soroush}).

In other experiments, we inject $K$ random link failures every 5 minute, with $K$ ranging from 1 to 7, more than some recent prior work (e.g., 1-2 failures in~\cite{teal_2023}). 


\parab{Metrics.}
For MLU-based problems, the solution quality is measured by the MLU.
For Max-Flow, solution quality is determined by \textit{Demand Satisfaction}, defined as the ratio of total routed flow to total demand.
However, since \sysname{} adapts online to demand changes and failures, but the baselines do not, we use a \textit{regret} metric that captures the performance of the solver over time (instead of at a single instant).

We define two kinds of regret: \textit{objective regret} and \textit{capacity regret}
(\cref{sec:form-defin-regr} describes these formally).
To understand objective regret, consider an MLU experiment in which we induce traffic changes every 20 seconds.
At each traffic change, we compute the instantaneous \textit{optimal} MLU, and compute, for \sysname{} and the baselines, the instantaneous MLU optimality gap (the difference between the baseline's MLU and the optimal).
For other baselines, this instantaneous gap remains the same until the next demand change (\eg after 20 seconds), or until when they recompute a new solution (e.g., after 5 minutes).
For \sysname{}, this instantaneous gap remains until it converges after a demand change, at which point it changes because \sysname{} will have computed a near-optimal allocation.
The objective regret is the integral (the area under the curve) of the instantaneous gaps.
We compute objective regret for Max-flow in the same way, with the appropriate definition of optimality gap.

To understand capacity regret, consider a Max-Flow solver.
When a demand increases, it may cause a capacity violation until the solver installs a new allocation.
To capture this, at each 20-second mark, we compute the amount by which MLU exceeds 1.
Capacity regret is the sum of these across the duration of the experiment.
These definitions are consistent with those used in convex optimization~\cite{lei-online-regret, mahdavi-online-regret}.


Finally, when evaluating \sysname{}'s edge-based solver (\cref{sec:eval-edge-based-solver}), we also report the \textit{average stretch} per commodity, defined as the sum of path lengths used by a commodity, weighted by the demand split routed on that path and normalized by the shortest path.

\begin{figure*}[t]
  \centering
  \includegraphics[width=0.9\textwidth]{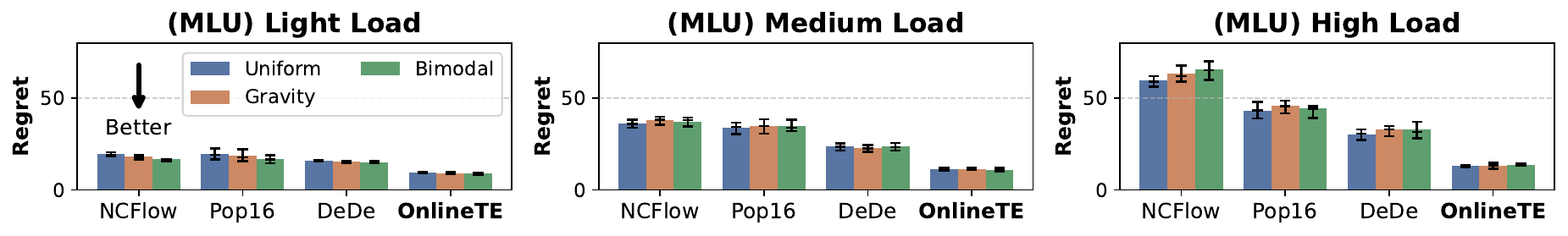}
  \caption{\label{fig:mlu_bar} Distribution of final MLU objective regret with different network load and demand distribution.}
\end{figure*}

\parab{Experiment setup.}
Each run of an experiment lasts 10 minutes.
For a given setting (traffic matrix type or link failure), we conduct 15 runs, and present metric value distributions (mean, median, p25 and p75) across those runs.

\subsection{Baseline Comparisons}\label{sec:baseline-comparisons}

We compare \sysname{} against baselines for demand changes and failures for path-based MLU and Max-flow.

\subsubsection{Demand Changes}\label{sec:demand-changes}

When demands change every 20 seconds, \sysname{} immediately starts re-optimizing its allocations.
Until it converges, it incurs regret.
If it does not converge before the next demand change, as sometimes happens, we warm-start \sysname{} with the updated demand at that instant.
Other approaches do not re-optimize, and therefore accumulate regret for the entire 5-minute interval. 
This section quantifies the differences between these approaches.



\parab{Path-based MLU.}
\cref{fig:mlu_bar} shows the distribution of objective regret across all baselines for the three different traffic loads on KDL.
For MLU, any capacity violation is already captured in the objective regret, so we do not need to compute capacity regret.
For each approach, the figure shows the mean (the height of the bar) and notches representing the median, the 25th percentile and the 75th percentile.

At all loads, \sysname{} has non-zero objective regret, for two reasons.
First, when a demand change occurs, it incurs regret until it converges, which takes a few seconds.
Second, we run ADMM until it reaches 1\% of the optimal.
Its objective regret increases slightly with load, and there is little variability across experiments.

DeDe computes an optimal allocation every 5 minutes.
For all demand changes that occur in the interim, its allocation can be suboptimal, so its objective regret is higher than \sysname{}'s (about 2$\times$ at low load, and up to 3$\times$ at higher loads).
Relative to DeDe, \sysname{} is able to lower objective regret by quickly reacting to demand changes.
\cref{fig:mlu_trace} depicts this graphically; DeDe's MLU changes when traffic demands change, but remains flat in between changes, except at the 5 minute mark when it recomputes its solution.
\sysname{} continuously re-optimizes after every demand change and is often able to converge before the next demand change (the flat sections at the bottom of the blue curve).
\begin{figure}[b]
  \centering
  \includegraphics[height=0.5\columnwidth]{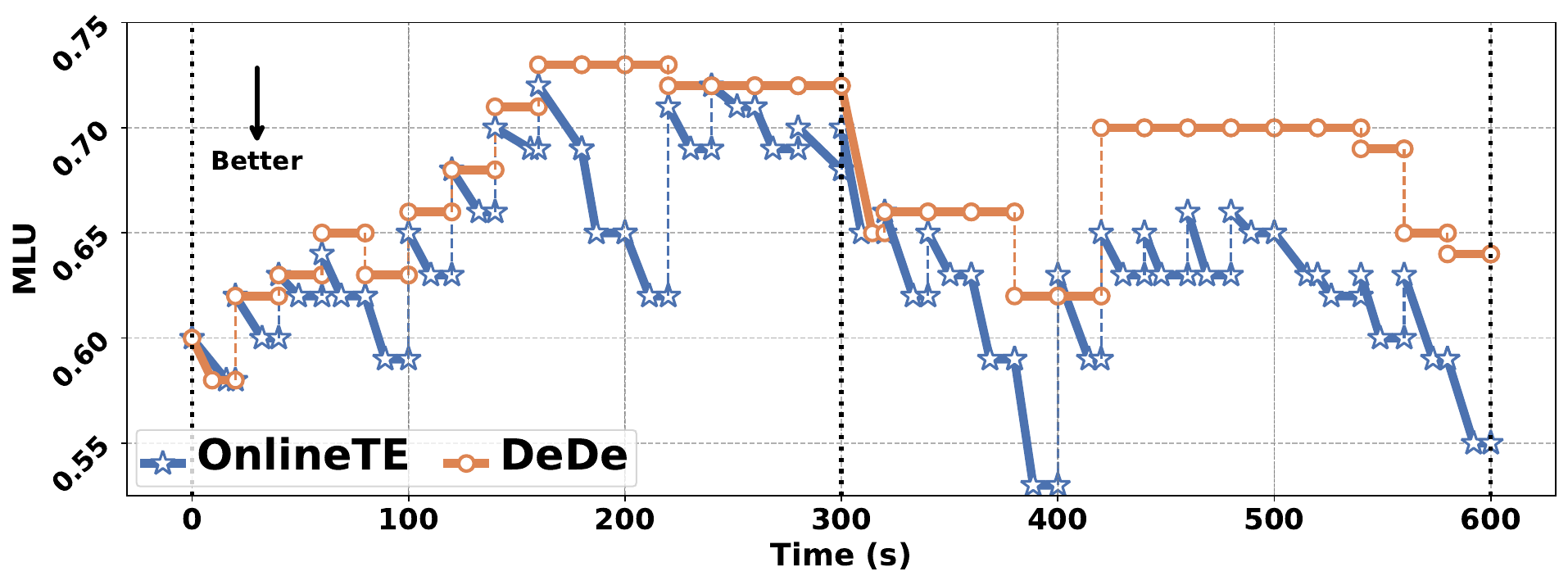}
  \caption{\label{fig:mlu_trace} Trace of MLU objective for one of our against DeDe.}
\end{figure}

NCFlow and POP use heuristics to determine allocation; they partition the problem space, solve sub-problems near-optimally and combine the solutions heuristically, resulting in faster, but sub-optimal, allocations.
At low traffic loads, they compare well with DeDe, but at higher traffic loads, their inherent sub-optimality results in regrets that may be 2$\times$ that of DeDe and up to 8$\times$ that of \sysname{} at high load.

The choice of traffic distribution has minimal impact on regret across traffic loads and approaches.
\sysname{} is almost insensitive to the traffic model, while at higher loads, NCFlow shows some variability.
These differences are likely due to the topological structure of KDL.

\begin{figure}[t]
  \centering
  \includegraphics[width=0.5\textwidth]{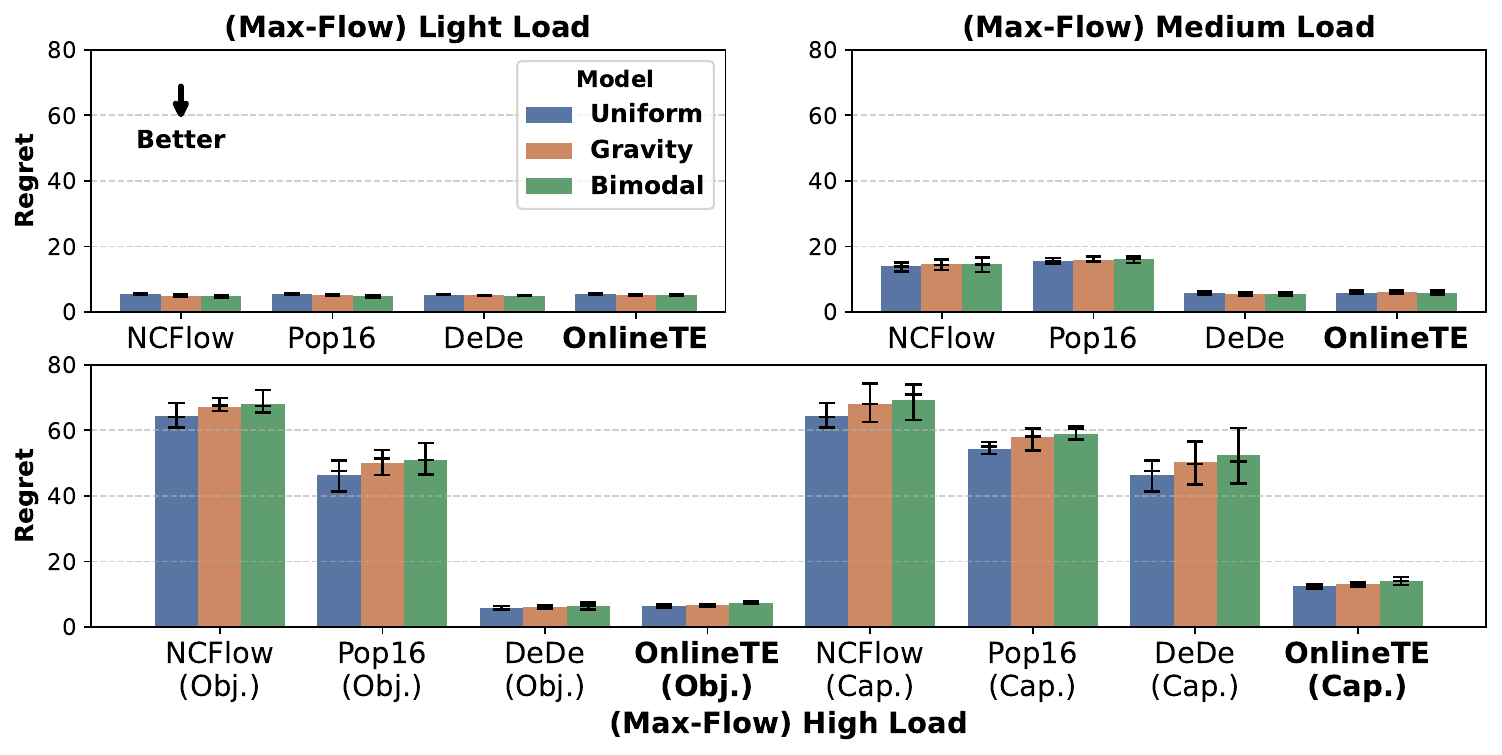}
  \caption{\label{fig:maxflow_bar} Distribution of final max-Flow objective regret.}
\end{figure}

\parab{Path-based Max-flow.}
\cref{fig:maxflow_bar} shows the regret for the Max-flow objective on KDL.
Across all traffic loads, \sysname{} has extremely low objective regret.
Interestingly, across low and medium traffic loads, \sysname{} and DeDe have comparable regrets.
For both, we compute solutions within 1\% of the optimal, but in theory, DeDe's objective regret should have been higher since it does not react immediately (as is the case for MLU at light and medium loads, \cref{fig:mlu_bar}).
It is comparable because the objective regret for Max-flow does not capture this transient behavior and does not account for capacity violations.
At low and medium loads, DeDe is near-optimal so it can route all demands, including the transients, so the optimality gap does not change when traffic demands change (as it does for MLU in \cref{fig:mlu_trace}).

At low loads, NCFlow and POP are comparable to \sysname{} because the network has enough capacity to route offered load even with heuristics.
At medium loads, their sub-optimality is evident, and their objective regret is 2-3$\times$ that of DeDe and \sysname{}.

At high-load, \sysname{}'s superior objective regret becomes more evident (left half of the bottom plot in \cref{fig:maxflow_bar}).
NCFlow's objective regret is more than 10$\times$ that of \sysname{}, and POP's is about 7$\times$.
But, DeDe's performance compares well with \sysname{}.
That is because the Max-flow objective regret does not penalize DeDe for not reacting to transients and for overloading links (as explained above).

The difference between the two shows up in \textit{capacity regret} (right half of bottom plot of \cref{fig:maxflow_bar}).
\sysname{} re-optimizes it solution in response to demand changes and incurs fewer capacity violations.
DeDe's capacity regret is about 4-5$\times$ that of \sysname{}.
On the other hand, at high loads, NCFlow and POP have very high capacity and objective regrets.

\subsubsection{Link Failures}\label{sec:link-failures}

In this section, we evaluate \sysname{}'s performance under failures for path-based MLU on KDL.
Every 20 seconds, we randomly fail a link with probability $p$.
We choose $p$ so that, on average, there are $m$ link failures, $m=1,3,5,7$, within a 5-minute interval.

When a link fails, a centralized TE solution usually invokes Fast Re-Route (FRR) to mitigate packet losses.
For example, if a link failure causes one or more of the pre-computed paths to fail, OneWAN~\cite{krishnaswamy23:_onewan} re-distributes its traffic among the remaining paths in proportion to the programmed splits on those paths.
We do the same for \sysname{}; it invokes FRR, then immediately initiates a re-optimization.

We compare this against an approach that starts with the optimal, and simply applies FRR but does not re-optimize.
\cref{fig:linkfail_mlu}(lower) shows the behavior of the two approaches during one of our experiments in which we induced 3 failures (at 60s, 180s and 200s).
In these experiments, we do not introduce any demand changes, since it would be difficult to dis-entangle \sysname{}'s behavior under failure from its response to demand changes.
These experiments use medium load and a uniform traffic matrix.

\cref{fig:linkfail_mlu}(upper) shows the objective regret for MLU as a function of the number of failures.
For 1-3 failures, FRR has 4-10$\times$ the regret of \sysname{}.
With more failures, the gap narrows slightly because \sysname{} can take longer to recover when a failure completely invalidates the prices (\textit{i.e.} dual variables) on the failed path.
In particular, if the failed link was heavily utilized, \sysname{} can take several iterations to converge to the optimal.
This is evident in \cref{fig:linkfail_mlu}(lower); after the first and third failures, \sysname{} takes 30-40 seconds to converge.

We expect to be able to lower regret with some optimizations.
For example, instead of using the FRR strategy before re-optimizing, \sysname{} can split traffic based on ADMM link prices.
Since these indicate congestion, \sysname{} could spread traffic to avoid congesting already congested links.
Even without this optimization and others which can lower convergence delays, which we defer to future work, \sysname{} clearly outperforms FRR in these failure regimes.

\begin{figure}[t]
    \centering
    \subfigure{
    \includegraphics{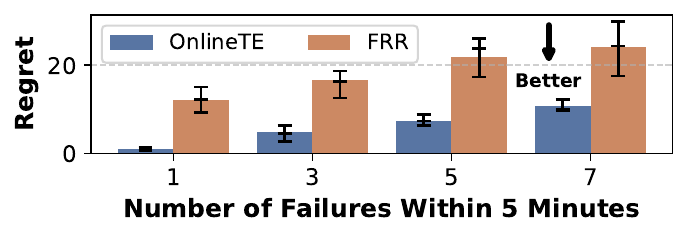}
    } \\
    \subfigure{\includegraphics{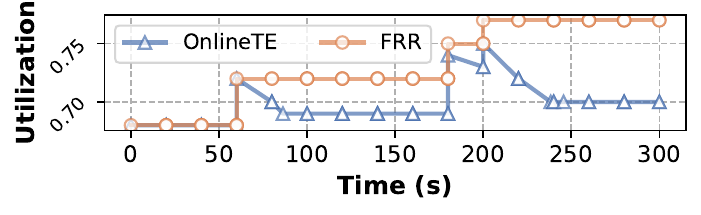}}
    \caption{Path-based MLU in presence of link failures. Upper figure: Objective regret. Lower figure: MLU over time with 3 failures at 60, 180 and 200 seconds.}
    \label{fig:linkfail_mlu}
\end{figure}
%


\subsection{Edge-Based Solver and Sparsity}\label{sec:eval-edge-based-solver}

A major contribution of \sysname{} is enabling edge-based TE optimization at scale.
In this section, we demonstrate through experiments that edge-based TE can significantly outperform path-based TE, especially under high-load.

For this experiment, we run both edge-based TE and path-based TE on KDL.
For edge-based TE, we sweep over 7 different $\kappa$ values from $1e^{-10}$ to $5e^{-07}$; $\kappa$ controls the sparsity of the allocation (\cref{sec:edge-based-solvers}), and higher $\kappa$ values constrain the demands to fewer paths, thereby resulting in lower \textit{stretch}.

\cref{fig:sparse_pareto} depicts the results.
In this figure, each dot or cross represents one experiment, either of \sysname{}'s edge-based solver at a $\kappa$ value indicated by the corresponding color (see legend) or its path-based solver.
A cross indicates the MLU for the median stretch, and the dot the MLU at the p95 stretch from that experiment.
Each large circle represents the average of all the p95 stretches and their corresponding MLUs across all experiments for the same value of $\kappa$ (or across all path-based solutions).
We then fit a Pareto frontier, shown in \cref{fig:sparse_pareto}, across these large circles; this represents the trade-off between latency and MLU on KDL for the edge-based solver at high load.

The figure dramatically illustrates the benefit of edge-based solvers.
The path-based solver, at high-load, exhibits an MLU of about 1.1, with a stretch of about 1.8.
This is well-beyond the Pareto frontier of \sysname{}'s edge-based solver.
\sysname{} can achieve a feasible high MLU of about 0.985 with a p95 stretch of 1.5, or 0.965 MLU with a slightly higher p95 stretch than path-based approaches.

While more extensive experiments can quantify the differences between edge-based and path-based TE, and how these depend on offered load and topology, \cref{fig:sparse_pareto} already demonstrates that (a) path-based TE can be far from the Pareto frontier of edge-based TE, and (b) that \sysname{} can solve for edge-based TE the scale of WANs with 750 nodes.




\begin{figure}
    \centering
    \includegraphics[width=0.9\columnwidth]{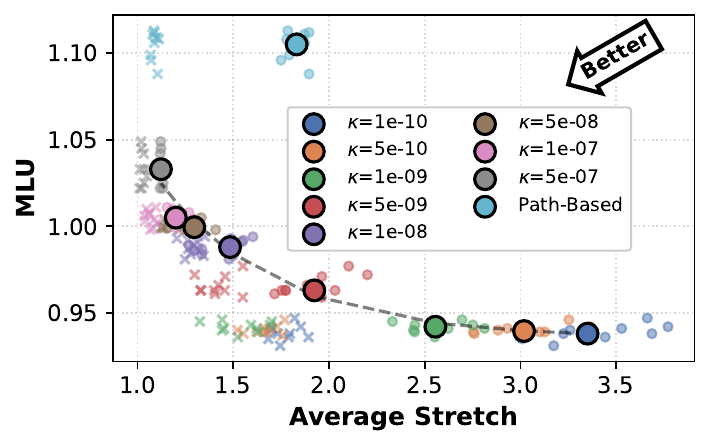}
    \caption{Edge-Based \sysname{} with different sparsity coefficients. We measure the MLU and average delay median and p95. Cross-marks indicate the median and the dashed line indicates the approximate Pareto frontier of tail latency against utilization.}
    \label{fig:sparse_pareto}
\end{figure}




\subsection{The Importance of Hierarchy}
\label{sec:eval-hier}


To demonstrate the importance of hierarchical ADMM (\cref{fig:hierarchy}), we compare \sysname{} against (a) a purely synchronous ADMM and (b) an asynchronous ADMM both with no hierarchy.
\sysname{} uses synchronous ADMM within regions, and asynchronous ADMM between them.


{
\setlength{\abovecaptionskip}{0pt}
\begin{figure}[b]
    \centering
    \includegraphics[width=\columnwidth]{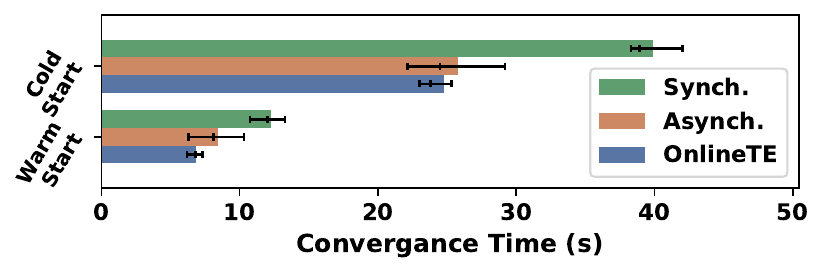}
    \vspace{-2pt}
    \caption{Convergence time for \sysname{} compared to synchronous and asynchronous ADMM \label{fig:conv_dist}}
\end{figure}
}
For this experiment, we emulate a inter-continental WAN topology, Cogentco (\cref{fig:topologies}), since its large geographic footprint illustrates the importance of hierarchy better.
For simplicity, we configure two regions (one in the US and one in Europe, \cref{fig:topologies}), with about 100 ms roundtrip time between them.
We use the path-based MLU solver for this experiment.

\cref{fig:conv_dist} shows the convergence time of the three approaches, both for cold-start and warm-start.
For both, synchronous ADMM has very high latency, almost 40~s for cold-start, since the coordinator needs to wait to hear from all switches.
Asynchronous ADMM converges faster, but can incur long convergence due to a large update from a far away switch.
As a result, it exhibits high variability (over 10s warm-start at p75).
\sysname{}'s hierarchical ADMM dampens stale updates from domains by only passing relatively stable aggregate updates between domains, resulting in faster warm-start (about 7s at p75).
%

\subsection{Switch Computation Overhead}\label{sec:switch-comp-overh}

The \xiupd{} of \cref{fig:nested_admm} runs projected gradient descent~\cite{boyd_co} on switches.
Our implementation parallelizes these across demands.
In our experiments, all of our switch VMs use only 2 CPU cores at $2.0$ GHz  and 2 GB DRAM, less resources than those available on modern switches (\cref{sec:methodology}).
The \xiupd{} does not require significant memory.

\cref{fig:switch_runtime} shows the execution time of a single invocation of \xiupd{} on our VMs as a function of topology size, both for edge-based and path-based approaches.
With a topology with $m$ nodes and $n$ edges, the \xiupd{} for edge-based TE  handles matrices of size $O(mn)$.
A path-based solver can use a more efficient sparse-matrix representation, since the paths are pre-defined and the allocations can be sparse.
As the topology size increases, the computation times increase approximately linearly (both axes are log-scale in \cref{fig:switch_runtime}).
At the size of KDL, an edge-based iteration incurs about 200~ms, and a path-based one about 10-20~ms.
However, the edge-based solver incurs few iterations before converging, so the actual difference in convergence between the two is less.
In practice, we can obtain linear speedup by using more cores (the Arista 7818 has 6) when necessary.



\begin{figure}
    \centering
    \includegraphics[width=0.9\columnwidth]{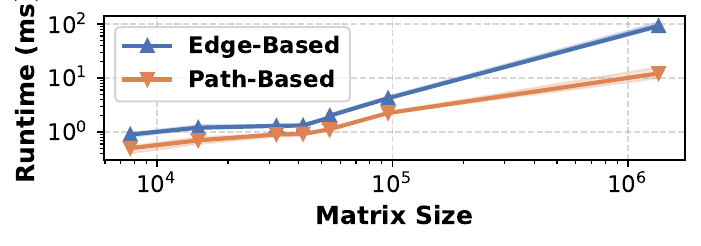}
    \caption{Overhead of a single switch-level iteration for the inner loop of \sysname{}.}
    \label{fig:switch_runtime}
\end{figure}



\section{Related Work}\label{sec:related-work}



\cref{sec:centralized-te}, \cref{sec:distributed-te} and \cref{tab:comparison} have extensively covered the most relevant related work in centralized and distributed TE, and explain how \sysname{} differs from these.
In the rest of this section, we cover other related work.

Beyond DOTE and Teal, there have been several papers on ML-based TE, including HARP~\cite{HARP} and FigRET~\cite{figret}.
Unlike \sysname{}, they are not mathematically guaranteed to be optimal.
RedTE~\cite{redTE} reacts to demand burst at short timescales by using a learned ML model to adjust splits locally until a global TE solver re-computes splits.
It is sub-optimal and complementary to \sysname{} since it targets burst at finer time-scales that the convergence time of \sysname{}.
MegaTE~\cite{megaTE} offloads some of the TE problem onto endpoints which may not be applicable in an ISP setting.


Tangentially related to our work is DATE~\cite{date}, a decentralized joint optimization of routing and congestion control, which might not be practicable at today's scales.
Finally, optimization decomposition has been applied in various domains and a full survey is beyond the scope of this paper; see~\cite{boyd11:_distr} for references.
Closest to \sysname{}, NUM-LAO~\cite{num-lao} views layering as optimization decomposition and \emph{implicitly} implements distributed algorithms by inferring costs from measurable states, e.g., packet drops.
\sysname{} leverages switch compute to \emph{explicitly} implement a distributed optimization problem. 

\section{Conclusions and Future Work}\label{sec:concl-future-work}

\sysname{} is a qualitatively different approach to WAN traffic engineering, in that it is near-optimal, can react to failures and demand changes in seconds, and can support more efficient edge-based TE. On the largest publicly available WAN, its objective regret is sometimes an order of magnitude better than the state of the art. Future work can explore how \sysname{} scales, how its edge-based solver performs on other topologies, and how to further optimize its convergence.

\newpage
\bibliographystyle{ACM-Reference-Format}
\bibliography{references.bib}

\newpage
\appendix
\renewcommand\thefigure{A.\arabic{figure}}
\renewcommand\thesection{\Alph{section}}
\renewcommand{\thealgocf}{A.\arabic{algocf}}
\renewcommand\thetable{A.\arabic{table}}
\setcounter{section}{0}
\setcounter{algocf}{0}
\setcounter{figure}{0}
\setcounter{table}{0}

\section{Formal Definitions of Regret}\label{sec:form-defin-regr}

We measure the \textit{regret} of \sysname{} and the baselines in two ways:
\begin{itemize}
    \item \textbf{Objective Regret} is accumulated sum of objective sub-optimality. Let $O_t$ denote the objective at time $t$ and $O_t^{*}$ be the optimal value; for a problem where the objective is to \textit{minimize} \footnote{For the maximization case, subtract the current objective from the optimal value instead} a value, we have:
    \begin{align*}
        \text{Objective Regret (T)} \coloneq \sum_{t=0}^{t=T-1} \max\big(0, O_t - O_t^{*}\big)
    \end{align*}
    \item \textbf{Capacity Regret} is accumulated sum of link capacity violation. Let $U_t$ be the solution MLU at iteration $t$, then we have:
    \begin{align*}
        \text{Capacity Regret (T)} \coloneq \sum_{t=0}^{t=T-1} \max\big(0, U_t - 1\big)
    \end{align*}
\end{itemize}
This notion of having separate regrets for objective and constraints has been used in convex constrained online optimization literature~\cite{lei-online-regret, mahdavi-online-regret}.
We only focus on capacity constraint violation since these are the only constraints that might be violated by a demand change in a path-based setting.

\section{Edge-Based Formulations}\label{app:edge-based-form}
%
{
\begin{table}[h]
\centering
\renewcommand{\arraystretch}{1.3}
\small
\begin{tabularx}{\columnwidth}{lX}
\toprule
\textbf{Term} & \textbf{Definition} \\ \midrule
\rowcolor{tablegray}
$G(V, E)$ & Directed, simple graph of topology with $|V|=m$ nodes and $|E|=n$ edges. \\
$(s_k, \epsilon_k, d_k)$ & Triplet defining commodity $k$, starting from source node $s_k$, ending in $e_k$ with demand $d_k > 0$. \\
\rowcolor{tablegray}
$C_e$ & Capacity of edge $e$. \\
$\mathcal{S}(k), \mathcal{E}(k)$ & Mappings returning the start and end node of a commodity $k$. \\
\rowcolor{tablegray}
\edgein{v}, \edgeout{v} & In-coming and out-going edge set of node $v \in V$. \\
$f_e(.)$ & Convex cost function of routing over edge $e$. \\
\rowcolor{tablegray}
$K$ & Total number of commodities\footnotemark \\
\bottomrule
\end{tabularx}
\caption{Notation for describing general TE problems.}
\label{tab:te-notation}
\end{table}
\footnotetext{One may assume $K = m(m-1)$, corresponding to each node sending and receiving from all other nodes.}
}
{
\begin{table}[ht]
\centering
\renewcommand{\arraystretch}{1.3}
\small
\begin{tabularx}{\columnwidth}{lX}
\toprule
\textbf{Term} & \textbf{Definition} \\ \midrule
\rowcolor{tablegray}
$X_{ek}$ & Assignment of routed demand from commodity $k$ on edge $e$. We have $X \in \mathbb{R}^{n\times K}$ \\
$B \in \mathbb{R}^{m \times K}$ & Demand matrix as in \cref{def:demand_vec}. \\
\rowcolor{tablegray}
$M \in \mathbb{R}^{m \times n}$ & Topology matrix as in \cref{def:graph_M}. \\
$f_{ek}(.)$ & Commodity-edge cost for routing demand $k$ over edge $e$. Often used for regularization.\\
\rowcolor{tablegray}
$\rho,~\eta$ & Outer and inner ADMM loop step sizes. \\
$\gamma$ & Switch problem step size. \\
\rowcolor{tablegray}
$\epsilon$ & $\ell_1$ or $\ell_2$ regularization coefficient. \\
$Z \in \mathbb{R}^n, ~P \in \mathbb{R}^{n \times K}$ & Outer and inner consensus variables. \\
\rowcolor{tablegray}
$T$ & Null-space dimension of $M$. \\
$\mathcal{N} \in \mathbb{R}^{n \times T}$ & Orthonrmal null-space basis of $M$.\\
\rowcolor{tablegray}
$\atiter{X}{0} \in \mathbb{R}^{n \times K}$ & Initial assignment satisfying \cref{eq:initial_sol}. \\
$U$ & Maximum link utilization across all edges. \\
\rowcolor{tablegray}
$r, u \in \mathbb{R}^n$ & Outer and inner dual variables. \\
$\lambda \in \mathbb{R}^{n \times K}$ & Switch problem dual variable ($\ell_2$ solver specific). \\
\rowcolor{tablegray} 
$S \in \mathbb{R}^{n \times K}$ & Switch problem slack variable ($\ell_1$ solver specific). \\
$\textbf{Prox}_{\alpha f}(v)$ & Proximal operator of $\alpha f(.)$. \\
\bottomrule
\end{tabularx}
\caption{Notation for our edge-based formulation.}
\label{tab:edge-notation}
\end{table}
}

We first describe the high-level optimization for edge-based TE.
%
%
\begin{definition}
\label{def:demand_vec}
For each commodity $k$, the demand vector $B_k \in \mathbb{R}^m$, is defined as:
\begin{align*}
    B_{kv} = \begin{cases}
        +d_k & v = s_k \\
        -d_k & v = \epsilon_k \\
        0 & \text{O.W.}
    \end{cases}
\end{align*}
\end{definition}
We also need a compact representation of the topology. For that, we can use the following:
\begin{definition}
\label{def:graph_M}
The node-edge incidence matrix $M \in \mathbb{R}^{m \times n}$:
\begin{align}
M_{ij} = \begin{cases}
    +1 & \text{if edge $j$ \textit{leaves} node $i$} \\
    -1 & \text{if edge $j$ \textit{enters} node $i$} \\
    0 & \text{O.W.}
\end{cases}
\end{align}
\end{definition}

An \textit{edge-based TE problem} is defined as follows:
\begin{alignat*}{2}
    \text{minimize} & \quad \sum\limits_{e=1}^n f_e\Big(\sum_k X_{ek}\Big) + \sum\limits_{e=1}^n && \sum\limits_{k=1}^K f_{ek}\Big(X_{ek}\Big)\\
    \text{s.t.}
        & \quad \forall k: MX_k = B_k &&\quad \text{Flow Conservation and} \\
        &&&\quad\text{Demand Constraint}\\
        & \quad \sum_k X_{k} \preccurlyeq C &&\quad \text{Edge Capacity Constr.} \\
        & \quad \forall k: 0 \preccurlyeq X_k &&\quad \text{Non-negativity} \\
        & \quad \forall k, e \in \text{\edgein{\mathcal{S}(k)}}: X_{ek} = 0 &&\quad \text{No demand loops} \\
        & \quad \forall k, e \in \text{\edgeout{\mathcal{T}(k)}}: X_{ek} = 0 &&\quad \text{No demand leaks} \\
        & \quad \text{(Other constraints)} &&\quad \text{TBD.}
\end{alignat*}
The first constraint is worth expanding;
it can be expressed more clearly as:
\begin{align}
\label{eq:flow-conservation}
\sum\limits_{e \in \text{\edgeout{v}}} X_{ek} - 
\sum\limits_{e \in \text{\edgein{v}}} X_{ek} = B_{vk} \quad 1 \leq v \leq m
\end{align}
Essentially, it means that on transit nodes (where $B_k$ entries are zero), we have flow conservation.
On source and destination nodes, the difference of the out-going and in-coming flows must equal $+/-$ the demand respectively.
This also illustrates why the loop/leak constraints are needed, as the above allows for cases where a source can produce more than the demand and offset by receiving the same demand later (this applies to destinations consuming demands as well).
We elect to make this constraint explicit, as loops within networks are devastating and we have empirically observed that these constraints aid convergance.

The edge-based assignment has been seldom tackled in context of practical TE; the price payed for optimality is offset by the many variables to be tuned and practical considerations.
Despite this, as we shall see, it lends itself very cleanly to a distributed solver.

We will begin with MLU. It can be written as:
\begin{alignat*}{2}
    \text{minimize} 
        & \quad g(U) + 
        \sum\limits_{e=1}^n f_e\Big(\sum_k X_{ek}\Big) + \sum\limits_{e=1}^n \sum\limits_{k=1}^K f_{ek}\Big(X_{ek}\Big)\\
        \text{s.t.}
            & \quad \forall k: MX_k = B_k \\
            & \quad \sum_k X_{k} \preccurlyeq UC  \\
            & \quad \forall k: 0 \preccurlyeq X_k \\
            & \quad \forall k, e \in \text{\edgein{\mathcal{S}(k)}}: X_{ek} = 0 \\
            & \quad \forall k, e \in \text{\edgeout{\mathcal{T}(k)}}: X_{ek} = 0 \\
            & \quad 0 \leq U
\end{alignat*}
The function $g(.)$ is a linear function in terms of the utilization.
The MLU variable $U$ should ideally be less than 1, but making that constraint explicit would generate a problem that can be completely infeasible at times, thus $U$ has no upper bound in this formulation.

\subsection{Transit Loops in Edge-Based Solutions}
Before we describe how to decompose the problem, we must emphasize that the edge-based formulation requires a bit of care to make sure transit loops (\ie loops over vertices other than the source or destination) do not happen.
A feasible solution can be easily checked for lacking loops by looking at the sub-graph induced by the support of $X_k$ (\ie all edges with positive flow assignments) and verifying that it is acyclic.
\begin{figure}[h]
  \centering
  \includegraphics[width=1.0\columnwidth]{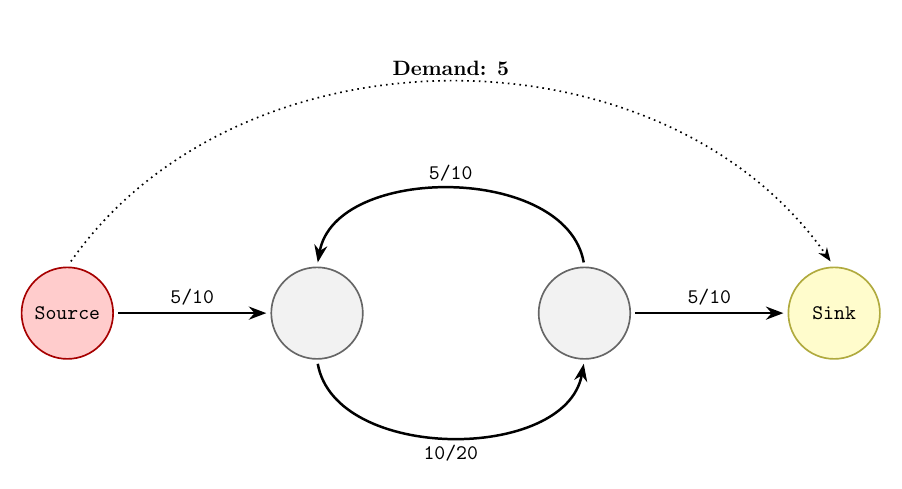}
  \caption{\label{fig:loop_in_mlu} An optimal edge-based MLU solution need not be free of loops if the cost depends only on the MLU. The notation \texttt{x/c} means one routes \texttt{x} flow units over a link with capacity \texttt{c}.}
\end{figure}
However, as-is nothing stops an edge-based solution from creating transit loops.
This is especially the case for MLU, where the cost is determined only by the most congested path, and assignments over non-congested links can easily contain loops.
\cref{fig:loop_in_mlu} shows a very simple example of this.

This is where our edge-commodity cost functions, $f_e$ and $f_{ek}$, come to help us.
\begin{theorem}
\label{th:acyclic_optimal}
If $f_e$ and $f_{ek}$ are convex, strictly increasing functions on $\mathbb{R}^+$, then the optimal edge-based solution $X^*$ cannot contain a loop.
\end{theorem}
\begin{proof}
This can be shown by contradiction. Since $X^*$ is feasible, then all of its entries are non-negative and do not create a loop that crosses the source or destination nodes.
As such, we need only focus on transit loops.
Assume one such loop exists for some commodity $k$, thus it must be possible for one to create a cyclic graph using edges from the support of $X_k^*$.
Let $\mathcal{C}$ denote the set of edges on this cycle.

Now, let $\nu$ be the smallest flow assignment for $k$ over the cycle $\mathcal{C}$.
Since $\mathcal{C}$ is a subset of the support of $X_k^*$, then $\nu > 0$.
The idea is to subtract $\nu$ from all assignments over this cycle.
Thus we can define:
\begin{align*}
    \tilde{X}_{k} = \begin{cases}
        X_{ek}^* - \nu & e \in \mathcal{C} \\
        X_{ek}^* & \text{O.W.} \\
    \end{cases}
\end{align*}
It is easy to see that $\tilde{X}_k$ is still a feasible assignment.
The demand loop and leak constraints remain satisfied, as $\tilde{X}_k$ agrees with $X_k^*$ on the associated edges.

Since all edges in $\mathcal{C}$ must carry at least $\nu$ units of flow, then the entries of $\tilde{X}_k$ are still non-negative.
Since we have also only decreased the total routed flow, the capacity constraint cannot be violated as the total flow on each edge is no larger than before.

Finally, flow conservation remains satisfied. We need only look at the expanded constraint in \cref{eq:flow-conservation} for any vertex that is part of the cycle $\mathcal{C}$, and note that both sums have been decremented by $\nu$ and thus the equality still holds.

The new solution $\tilde{X}_k$ also must contain at least one edge from $\mathcal{C}$ with a zero assignment (which previously routed exactly $\nu$ units), thus $\tilde{X}_k$ breaks the cycle be removing one of the edges.
We can now repeat this for all $k$ and get a full feasible solution $\tilde{X}$, which given the assumptions on $f_e$ and $f_{ek}$ \textit{must} have a smaller objective.
Thus $X^*$ could not have been the optimal solution which is a contradiction.

Therefore, no such loop like $\mathcal{C}$ can exist in the optimal solution for us to take as witness.
\end{proof}
The proof above remains sound if one were to remove the functions $f_e$.
The proof however breaks if one were to remove the functions $f_{ek}$ as well.
For this reason, these functions are not just for regularization, but are necessary in preventing loops in the final edge-based solution.

\subsection{The ADMM Decomposition}

One source of difficulty across both edge-based and path-based formulations of all objectives is the capacity constraint.
This constraint couples every individual assignment across commodities, and as such, if we could push it away into somewhere else, things will become much easier.
ADMM is a reasonable choice for this (note that related work TEAL has also observed this~\cite{teal_2023}).

To do this, introduce $Z_{e} := \sum_k X_{ek}$ as the aggregate flow over edge $e$ across commodities. We can rewrite the problem as:
\begin{alignat*}{2}
    \text{minimize} 
        & \quad g(U) + \sum_e f_e\Big(Z_e\Big) + \sum_e \sum_k f_{ek}\Big(X_{ek}\Big)\\
        \text{s.t.}
            & \quad \sum_k X_k - Z = 0 \\
            & \quad \mathcal{F}_X ~\text{and}~ \mathcal{F}_Z ~\text{hold}.
\end{alignat*}
Where the constraints over $X$ are defined as:
\begin{alignat*}{2}
    \mathcal{F}_X = \begin{cases}
        & \forall k: MX_k = B_k \\
        & \forall k: 0 \preccurlyeq X_k \\
        & \forall k, e \in \text{\edgein{\mathcal{S}(k)}}: X_{ek} = 0 \\
        & \forall k, e \in \text{\edgeout{\mathcal{T}(k)}}: X_{ek} = 0
    \end{cases}
\end{alignat*}
And over $Z$:
\begin{alignat*}{2}
    \mathcal{F}_Z = \begin{cases}
        & 0 \leq U \\    
        & Z \preccurlyeq uC
    \end{cases}
\end{alignat*}

Note that besides the consensus constraint, the other constraints separate into distinct sets, one of which only depends on $Z$ and $u$, and the other on $X$.
In this format, we can write this as an ADMM loop by creating the augmented Lagrangian with the first constraint and distributing the two constraint sets over appropriate updates.
Let $\indexedatiter{X}{k}{m}$ be the value of $X_k$ during ADMM iteration $m$ (the same applies for $Z$ and $r$).

To do this, introduce the dual variable $r \in \mathbb{R}^n$ for that constraint, and we get the following steps:
\begin{alignat*}{2}
    \textbf{X-update:}&\quad \\
    \atiter{X}{m+1} &:= \argmin_{\mathcal{F}_X}\; \sum_e \sum_k f_{ek} (X_{ek}) + \\
    &\quad~~\frac{\rho}{2} \Big\| \sum_k X_k - \atiter{Z}{m} + \atiter{r}{m} \Big\|_2^2 \\
    \textbf{Z-update:}&\quad \\
    \atiter{Z}{m+1} &:= \argmin_{\mathcal{F}_Z}\; g(U) + \sum_e f_e (Z_e) + \\
    &\quad~~\frac{\rho}{2} \Big\| \sum_k \indexedatiter{X}{k}{m+1} - Z + \atiter{r}{m} \Big\|_2^2 \\
    \textbf{Dual Update:}&\quad \\
    \atiter{r}{m+1} &:= \atiter{r}{m} + (\sum_k \indexedatiter{X}{k}{m+1} - \atiter{Z}{m+1})
\end{alignat*}
The above is the scaled ADMM form~\cite{boyd11:_distr} and $\rho$ is a step size that must be carefully chosen.
This constitutes the \textit{Outer Loop} of our algorithm, where assignments are made and the matrix $X$ decided upon, then fine-tuning on aggregate flow is done such that the utilization is lowered.

The update on $Z$ is simple, as it only has variables in number of edges.
The $X$ update however is still a monster due to how many variables it has; we note however that besides the linear constraint, the other constraints are \textit{simple} and can be handled with projection, \textit{and} they also separate over $k$.

We must now scale the problem such that it can be solved quickly.
Assuming $f_{ek}$ isn't an extremely complicated function, the $X$ update is mostly defined by its quadratic objective and the non-negativity constraint.
However, if we shuffle the summation on $f_{ek}$, we see an important structure.
\begin{alignat*}{2}
    \text{minimize}\quad& 
        \sum_k \Big( \sum_e \ f_{ek} (X_{ek}) \Big) + \frac{\rho}{2} \Big\| \sum_k X_k - Z + \atiter{r}{m} \Big\|_2^2 \ \\
    \text{s.t. ...}
\end{alignat*}
From now on, we define $\atiter{F}{m} := \atiter{Z}{m} - \atiter{r}{m}$. Define the following as well:
\begin{alignat*}{2}
    \tilde{f}_{k}(X_{k}) := \sum_e f_{ek} (X_{ek}) \quad \text{and} \quad
    \tilde{g}(X_k) := \frac{\rho}{2} \Big\|  X_k - \atiter{F}{m} \Big\|_2^2
\end{alignat*}
The problem is then:
\begin{alignat*}{2}
    \text{minimize}&\quad 
        \sum_k \tilde{f}_{k}(X_{k}) + \tilde{g}(\sum_kX_k) \\
        \text{s.t. ...}
\end{alignat*}
A creature of this structure is more commonly known as a \textit{Sharing} problem~\cite{boyd11:_distr}.
Typically, such problems are unconstrained, but here the existence of the non-negativity constraint is not too bad, since it too separates over commodities.
We are quite fortunate, since ADMM just happens to lend itself very well to solving a problem like this. If we introduce the pairing $P_k$ for $X_k$ such that $P_k - X_k = 0$, (scaled) ADMM gives the following update steps:
\begin{alignat*}{2}
X_k^{(m+1)} &\coloneq \argmin_{\mathcal{F}_X} \sum_e f_{ek} (X_{ek}) + \\
&\quad~~~\frac{\eta}{2}\Big\| X_k - X_k^{(m)} + \bar{X}^{(m)} - \bar{P}^{(m)} + u^{(m)} \Big\|_2^2 \\
\bar{P}^{(m+1)} &\coloneq \argmin_{\bar{P}} \frac{\rho}{2} \Big\| K \bar{P} - \atiter{F}{m} \Big\|_2^2 + \\
&\quad~~~\frac{K^2\eta}{2}\Big|\Big| \bar{P} - u^{(m)} - \bar{X}^{(m+1)} \Big|\Big|_2^2\\
u^{(m+1)} &\coloneq u^{(m)} + \bar{X}^{(m+1)} - \bar{P}^{(m+1)}
\end{alignat*}
Where $\eta$ is the step size of this second ADMM loop, $\bar{P}$ and $\bar{X}$ are column-wise averages of their respective variables and $u$ is a dual variable.

The mean update has a closed form solution:
\begin{alignat*}{2}
\bar{P}^{(m+1)} \coloneq \frac{\frac{\atiter{F}{m}}{K} + \frac{\eta}{\rho}(u^{(m)} + \bar{Y}^{(m+1)})}{1 + \frac{\eta}{\rho}}
\end{alignat*}
Thus, we end up with:
\begin{alignat*}{2}
X_k^{(m+1)} &\coloneq \argmin_{\mathcal{F}_X} \sum_e f_{ek} (X_{ek}) + \\
&\quad~~~\frac{\eta}{2}\Big\| X_k - X_k^{(m)} + \bar{X}^{(m)} - \bar{P}^{(m)} + u^{(m)} \Big\|_2^2\\
\bar{P}^{(m+1)} &\coloneq \frac{\frac{\atiter{F}{m}}{K} + \frac{\eta}{\rho}(u^{(m)} + \bar{Y}^{(m+1)})}{1 + \frac{\eta}{\rho}} \\
u^{(m+1)} &\coloneq u^{(m)} + \bar{X}^{(m+1)} - \bar{P}^{(m+1)}
\end{alignat*}
Take note that we only need to keep $\atiter{\bar{P}}{m}$, not the individual $P_k$.
The above constitutes as our \textit{Inner ADMM Loop}, together with our Outer Loop, they form an intertwined procedure that nudges our assignment matrix closer to examples with less and less maximum utilization.
The outer loop provides an algorithmic benefit by decoupling aggregate flow from individual assignments, and the inner loop provides scalability by decoupling across commodities.

\subsection{A Useful Trick}
The previous decomposition (the outer and inner loop) are core to our suite of solvers, regardless of objective.
There is another trick that we also employ many times for edge-based solvers, and that concerns the linear constraint within $\mathcal{F}_X$.
Recall that the update to $X_k$ finally takes the form of:
\begin{alignat*}{2}
\text{minimize}
&\quad
\sum_e f_{ek} (X_{ek}) + \\
&\quad~~~\frac{\eta}{2}\Big\| X_k - X_k^{(m)} + \bar{X}^{(m)} - \bar{P}^{(m)} + u^{(m)} \Big\|_2^2\\
\text{s.t.} &\quad MX_k = B_k \\
&\quad 0 \preccurlyeq X_k \\
&\quad \forall e \in \text{\edgein{\mathcal{S}(k)}}: X_{ek} = 0 \\
&\quad \forall e \in \text{\edgeout{\mathcal{T}(k)}}: X_{ek} = 0
\end{alignat*}
We can make the linear constraint implicit!
In particular, assume that some feasible solution to the linear constraint like $\indexedatiter{X}{k}{0}$ is known:
\begin{alignat}{2}
    \forall k: M\indexedatiter{X}{k}{0} = B_k
    \label{eq:initial_sol}
\end{alignat}

Let $\mathcal{N} \in \mathbb{R}^{n\times T}$ be a matrix whose columns form a basis for the null space of $M$ and $T$ is the null space dimension of $M$.
Such a matrix can be easily found with SVD (with the added bonus that we can have $\mathcal{N}^T\mathcal{N} = I$, which we also assume).
Then every solution to the linear constraint would be of the form:
\begin{align*}
    X_{ek} = \indexedatiter{X}{ek}{0} + \sum_t \mathcal{N}_{et} Y_{tk}
\end{align*}
We can substitute this into the $X$ update and get the following:
\begin{alignat*}{2}
    \text{minimize}
    &\quad
    \sum_e f_{ek} (X_{ek}^{(0)} + \mathcal{N}Y_k) + \\
    &\quad~~~\frac{\eta}{2}\Big\| \indexedatiter{X}{k}{0} + \mathcal{N}Y_k - X_k^{(m)} + \bar{X}^{(m)} - \bar{P}^{(m)} + u^{(m)} \Big\|_2^2
    \\
    &\quad 0 \preccurlyeq \indexedatiter{X}{k}{0} + \mathcal{N}Y_k \\
    &\quad \text{etc.}
\end{alignat*}
This problem is small and manageable, any algorithm like Projected Gradient Descent (PGD) or Coordinate Descent can be used to solve it.
\sysname{} by default uses PGD.

As a final note, the choice of the initial feasible solution $\indexedatiter{X}{k}{0}$ can be helpful.
We consider mostly two choices:
\begin{itemize}
    \item \textbf{Uniform routing on shortest paths}. This is the simplest, and generates sparse assignments that we can keep in small memory and do quick arithmetic on (in particular, the products $\mathcal{N}^T \indexedatiter{X}{k}{0}$ do appear in many formulations, which itself is actually quite dense no matter the sparsity pattern of $\indexedatiter{X}{k}{0}$).
    \item \textbf{Least squares solution}. In particular, let $\indexedatiter{X}{k}{0} \coloneq M^{\dagger}B_k$, where $M^\dagger$ is the Moore-Penrose pseudo-inverse of $M$. The result will be a dense matrix often, but it brings with itself many computational boons, the main one being that $\mathcal{N}^T\indexedatiter{X}{k}{0} = 0$ holds.
\end{itemize}
In what follows, we assume that the least squares solution is always used.

\subsection{$\ell_2$ Regularized Form}
We can let the functions $f_{ek}$ be small quadratic functions.
As long as their values remain comparatively small and obey the assumptions in \cref{th:acyclic_optimal}, their exact structure won't matter.

One simple case is to consider functions $f_{ek}(x) = \frac{\epsilon}{2}(x)^2$ for a small $\epsilon \geq 0$.
By expanding these $f_{ek}$ we can see that the update to $\indexedatiter{Y}{k}{m+1}$ turns to\footnote{If we didn't use the least-squares feasible assignment, then the expansion also generates a linear term in $Y_k$. This does not change the structure of the final problem and the same methods can be applied to it.}:
\begin{alignat*}{2}
    \text{minimize} 
        & \quad \frac{\epsilon}{2} \Big|\Big| Y_k \Big|\Big|_2^2 + \frac{\eta}{2}\Big|\Big| Y_k - Y_k^{(m)} + \bar{Y}^{(m)} - \bar{P}^{(m)} + u^{(m)} \Big|\Big|_2^2 \\
    \text{s.t.}
        & \quad \forall k: 0 \preccurlyeq \indexedatiter{X}{k}{0} + \mathcal{N}Y_k \\
        & \quad \forall k, e \in \mathcal{S}(k) \cup \mathcal{T}(k)~: 0 = \indexedatiter{X}{k}{0} + \mathcal{N}Y_k
\end{alignat*}
Or in a slightly more standard form:
\begin{alignat*}{2}
    \text{minimize} 
        & \quad \frac{1}{2} \Big|\Big| Y_k \Big|\Big|_2^2 - \frac{\eta}{\eta + \epsilon} Y_k^T (Y_k^{(m)} - \bar{Y}^{(m)} + \bar{P}^{(m)} - u^{(m)}) \\
    \text{s.t.}
        & \quad \forall k: 0 \preccurlyeq \indexedatiter{X}{k}{0} + \mathcal{N}Y_k \\
        & \quad \forall k, e \in \mathcal{S}(k) \cup \mathcal{T}(k)~: 0 = \indexedatiter{X}{k}{0} + \mathcal{N}Y_k
\end{alignat*}

Among the easier ways to solve this problem, PGD seems suitable, but the projection that has to be done to make the solution feasible can get quite expensive.

We can make the projections trivial by solving the dual instead.
If we introduce the dual variable $\lambda_k \in \mathbb{R}^n$ and exchange the objective sign to get another minimization problem, we end up with:
\begin{alignat*}{2}
    \text{minimize} 
        & \quad \frac{1}{2} \Big|\Big| \mathcal{N}^T\lambda_k + \indexedatiter{C}{k}{m} \Big|\Big|_2^2 + \lambda_k^T\indexedatiter{X}{k}{0} \\
    \text{s.t.}
        & \quad \forall k, e \notin \mathcal{S}(k) \cup \mathcal{T}(k)~: 0 \leq \lambda_{ek}
\end{alignat*}
Where for compactness, we let $\indexedatiter{C}{k}{m} := \frac{\eta}{\eta + \epsilon} (\indexedatiter{Y}{k}{m} - \atiter{\bar{Y}}{m} + \atiter{\bar{P}}{m} - \atiter{u}{m})$.
The KKT conditions give the following relation to retrieve the primal solution from the dual:
\begin{alignat*}{2}
    Y_k^* = \mathcal{N}^T\lambda_k^* + \indexedatiter{C}{k}{m}
\end{alignat*}
PGD is very well suited to solve this problem, especially because the projections needed are trivial.
When sparsity isn't needed and we want to retrieve the \textit{true} edge-based optimal, this is the method that we employ.

\subsection{$\ell_1$ Regularized Form}
The constraints $MX_k = B_k$ and the non-negativity constraints, gives a very under-determined system, and we are interested in the most sparse solutions.
This can be achieved by penalizing the problem with an $\ell_1$ norm cost.
The lowest hanging fruit here would be to let:
$$
f_{ek}(x) = \epsilon \; | x |
$$

It is easy to verify that the sharing problem has not changed, the update to $\bar{P}$ and the dual update remain the same as with the $L_2$ norm, but the updates to $Y$ change to the following:
\begin{alignat*}{2}
    \text{minimize}\quad& 
        \epsilon \; \Big\| X_k \Big\|_1 + \frac{\eta}{2}\Big\| Y_k - Y_k^{(m)} + \bar{Y}^{(m)} - \bar{P}^{(m)} + u^{(m)} \Big\|_2^2 \\
        \text{s.t.}\quad ...
\end{alignat*}
Which is a constrained LASSO. 
It is not the easiest thing to solve in general, but our saving grace is that the switch problems are usually quite small.
There are multiple ways to solve these (coordinate decent is often used, and we also use it as well), but ADMM can also be used again\footnote{Which solidifies the hunch that the writer is a one-trick pony.}:
\begin{alignat*}{2}
\indexedatiter{X}{k}{m, i+1} &:= \argmin_{0 \preccurlyeq X_k, etc.} 
\frac{\epsilon}{\eta}\Big\| X_k \Big\|_1 + \\
&\quad~~~\frac{\gamma}{2} \Big\| X_k - \indexedatiter{X}{k}{0} - \mathcal{N}\indexedatiter{Y}{k}{m, i} + \indexedatiter{t}{k}{i} \Big\|_2^2 \\
Y_k^{(m, i+1)} &:= \argmin \frac{1}{2} 
\Big\| Y_k - \indexedatiter{C}{k}{m} \Big\|_2^2 + 
\\
&\quad~~~\frac{\gamma}{2} \Big\| \indexedatiter{X}{k}{m, i+1} - \indexedatiter{X}{k}{0} - \mathcal{N}\indexedatiter{Y}{k}{m, i} + \indexedatiter{t}{k}{i} \Big\|_2^2 \\
\indexedatiter{t}{k}{m,i+1} &:= \indexedatiter{t}{k}{m,i} + (\indexedatiter{X}{k}{m,i+1} - \indexedatiter{X}{k}{0} - \mathcal{N}\indexedatiter{Y}{k}{m, i+1})
\end{alignat*}
The second step has the same solution as we discussed in the $L_2$ norm case. 
The first problem however has a non-smooth objective. If we write it in an unconstrained form using the identifier function for the positive orthant, we get:
\begin{alignat*}{2}
\indexedatiter{X}{k}{m, i+1} &:= \argmin \frac{\epsilon}{\eta\gamma}\Big\| X_k \Big\|_1 + \mathcal{I}_+(X_k) + \\
&\quad~~~\frac{1}{2} \Big\| X_k - \indexedatiter{X}{k}{0} - \mathcal{N}\indexedatiter{Y}{k}{m, i} + \indexedatiter{t}{k}{i} \Big\|_2^2
\end{alignat*}
This is exactly the Proximal Operator of the function $\frac{\epsilon}{\eta\gamma}\| x \|_1 + \mathcal{I}_+(x)$. Both of these functions are applied element-wise and are summative, thus the proximal composition rule applies to them. This gives the following:
\begin{alignat*}{2}
\indexedatiter{X}{k}{m, i+1} &:= \text{Prox}_{+} (\text{Prox}_{\epsilon/\eta\gamma ||.||_1} (\indexedatiter{X}{k}{0} + \mathcal{N}\indexedatiter{Y}{k}{m, i} + \indexedatiter{t}{k}{i}) )
\end{alignat*}
The proximal operator of $\ell_1$ norm is the Soft Thresholding operator $\mathcal{S}_{\epsilon/\eta\gamma}$~\cite{neal-boyd-proximal}, and the proximal operator of the non-negative orthant is just clipping the negative values to zero.
Thus, the final solution involves the following:
\begin{alignat*}{2}
\indexedatiter{X}{k}{m, i+1} &:= \Big[ \mathcal{S}_{\epsilon/\eta\gamma} (\indexedatiter{X}{k}{0} + \mathcal{N}\indexedatiter{Y}{k}{m, i} + \indexedatiter{t}{k}{m,i}) \Big]_+ \\
Y_k^{(m, i+1)} &:= \frac{\indexedatiter{C}{k}{m} - \gamma \; \mathcal{N}^T (\indexedatiter{X}{k}{0} + \indexedatiter{t}{k}{i} - \indexedatiter{X}{k}{m, i+1})}{1 + \gamma} \\
\indexedatiter{t}{k}{m,i+1} &:= \indexedatiter{t}{k}{m,i} + (\indexedatiter{X}{k}{m,i+1} - \indexedatiter{X}{k}{0} - \mathcal{N}\indexedatiter{Y}{k}{m, i+1})
\end{alignat*}
And $\indexedatiter{X}{k}{m, i+1}$ are sparse solutions that we can report directly to the controller.
When it converges, this method provides nice solutions but can be slow and iterates can potentially violate flow conservation.

We want to emphasize that the above is slightly abusing notation.
In practice, we don't just clip the entries of $\indexedatiter{X}{k}{m, i+1}$, but we also pin the ones on the loop/leak constraints to zero as well.
Finally, note that while we let $f_{ek}$ only contain a weakly convex function, we can still add an $\ell_2$ norm to it (like the case we had in the previous section) and solve it with the exact same method as above.
This is because adding a simple $\ell_2$ norm would only translate to shift of the bias in the quadratic objective, which still leaves it a quadratic.

\section{Path-Based Formulations}\label{app:path-based-form}
{
\begin{table}[h]
\centering
\renewcommand{\arraystretch}{1.3}
\small
\begin{tabularx}{\columnwidth}{lX}
\toprule
\textbf{Term} & \textbf{Definition} \\ \midrule
\rowcolor{tablegray}
$\mathbb{P}$ & Set of all paths usable for routing. \\
$T$ & Maximum number of paths that a single commodity can be split over. When the commodity is known, we may index paths with $1 \leq t \leq T$. \\
\rowcolor{tablegray}
$p_k$ & Actual number of paths available for a commodity $k$ (this may not be $T$, due to failures or topology restrictions). \\
$Y_{tk}$ & The split of commodity $k$ over path index $t$. \\
\rowcolor{tablegray}
$\mathcal{I}$ & Path dictionary as in \cref{def:path-dict} \\
$\mathcal{P}(e, k)$ & Set of all path indices crossing edge $e$ that can serve commodity $k$ as in \cref{def:path-edge-map} \\
\rowcolor{tablegray}
$\textbf{PS}(S)$ & The power-set of set $S$. \\
$\alpha \in \{0, 1\}^{K \times n \times T}$ & The path mask matrix as in \cref{def:path-mask}. \\
\bottomrule
\end{tabularx}
\caption{Notation for our path-based formulation.}
\label{tab:path-notation}
\end{table}
}


%
We first define how commodities can make use of available paths:
\begin{definition}
\label{def:path-dict}
The path dictionary $\mathcal{I}: \{1,..,K\} \times \{1,..,T\} \to \mathbb{P}$ is the mapping that given a commodity and path index, outputs the associated path in $\mathbb{P}$.
In the case where less than $T$ paths are available or a path has become unavailable, it returns a placeholder $\emptyset$ value.
\end{definition}
To handle capacity constraints, we also need to define which path crosses which edge:
\begin{definition}
\label{def:path-edge-map}
The mapping $\mathcal{P}: E\times\{1,..,K\}\to\textbf{PS}(\{1,..,T\})$ gives all path indices crossing an edge $e$ and (potentially) serving a commodity $k$ (note that there is at most $T$ such paths, thus the indexing $\mathcal{I}$ can be used to convert to the actual path).
\end{definition}

Instead of a pure path-based formulation, we opt to use a formulation that has an edge-based \textit{flare} to it.
In particular, note that given a path-based assignment, it is trivial to recover an edge-based assignment.
To formalize this, we define the path mask matrix:
\begin{definition}
\label{def:path-mask}
For each commodity $k$ define $\alpha^{(k)} \in \{0, 1\}^{n \times T}$ such that for each column $\indexedatiter{\alpha}{t}{k}$ we have:
\begin{alignat*}{2}
\forall e \in E: \indexedatiter{\alpha}{et}{k} =
\begin{cases}
1& t \in \mathcal{P}(k, e) \\
0& \text{O.W.}
\end{cases}
\end{alignat*}
\end{definition}
It is easy to see that an edge-based assignment $X$ can be created from a path-based assignment $Y$ with:
\begin{alignat*}{2}
    \forall k: X_k = \alpha^{(k)} Y_k d_k
\end{alignat*}
By adding this constraint into any edge-based problem, a path-based equivalent is created.
There is one small change compared to the previous edge-based solver, where we may elect to introduce functions $f_{tk}$ indexed per-path, per-commodity instead of per-edge, per-commodity.
If we desire strong convexity, these functions are enough for that purpose, since a path-based assignment uniquely defines an edge-based assignment.
In such a setting, being loop-free follows from the paths in $\mathcal{P}$ being loop-free and have no relation to what choice of $f_{tk}$ or $f_e$ we make.

It may appear vexing that we would combine edge-based and path-based problems like this, but it has two main advantages:
\begin{itemize}
    \item It simplifies the capacity constraint greatly, allowing us to recover our outer and inner ADMM loops with hardly any difference compared to the edge-based case.
    From an implementation point of view, we may use any existing code that implements our edge-based backend for our path-based solvers.
    \item It pushes any consideration about paths and their availability, into the switches alone, sparing the controller from having to contend with it.
\end{itemize}
As such, a path-based solver is a very intuitive restriction of our edge-based solver.
Essentially, it is the exact same solver, but for each iteration, the switches project their edge-based solution onto the available paths.
For this reason, the path-based algorithm has the exact same message complexity as the edge-based solver.
The main benefit comes from the computation complexity, as the $Y_{tk}$ matrix is much smaller than the $X_{ek}$ matrix that we had to use previously and we also have no possibility of a loop (unless a configured path has a loop, which we assume never happens).
\subsection{Path-Based PGD Solution}
For the path-based solver, only the switch problem changes.
It is very easy to show that the problem would become:
\begin{alignat*}{2}
    \text{minimize}& \quad \sum_t f_{tk}\Big(d_kY_{tk}\Big) + \\
    &\quad\frac{\eta}{2} \Big\| \alpha^{(k)}Y_k d_k - X_k^{(m)} + \bar{X}^{(m)} - \bar{P}^{(m)} + u^{(m)} \Big\|_2^2 \\
    \text{s.t.}
        & \quad \sum_t Y_{tk} = 1 \\
        & \quad \forall  t>p_k: Y_{tk} = 0 \\
        & \quad 0 \preccurlyeq Y_{k}
\end{alignat*}
Let $\atiter{A}{k} := \Big( (\atiter{\alpha}{k})^T \atiter{\alpha}{k} \Big)$ be the $T \times T$ Gram matrix of the path mask.
How easy we can solve this problem, depends completely on what $f_{tk}$ is and what paths we use.
If $f_{tk}$ is smooth, then the gradient is:
\begin{alignat*}{2}
    \nabla_k =& d_k f'\Big(d_k Y_{tk}\Big) + \eta d_k^2
    \atiter{A}{k} Y_k - \\
    &\eta d_k (\atiter{\alpha}{k})^T \Big( \indexedatiter{X}{k}{m} - \atiter{\bar{X}}{m} + \atiter{\bar{P}}{m} - \atiter{u}{m} \Big)
\end{alignat*}
And we can proceed as before.
The projection step is projection onto a \textit{reduced} probability simplex, where some of the entries are pinned to zero beforehand.
This is mathematically equivalent to just projecting onto the Simplex in terms of the algorithm, but there are some small practical considerations that we won't discuss here.

However, we will take time to emphasize a certain case where the solution becomes much simpler.
That is the case of \textit{edge-disjointed} paths, meaning that for each commodity, none of the at most $T$ paths share any edges.

Equivalently, this means that the Gram matrix $\atiter{A}{k}$ is actually \textit{diagonal}.
With this, all previous formulations simplify greatly, however, it is understandable that we may be hesitant to limit ourselves to such a path selection.
An almost equally efficient way of handling the general case is possible by taking advantage of two things about $\atiter{\alpha}{k}$, its \textit{sparsity} and being Boolean valued.

Essentially, in order to implement the multiplications with $\atiter{\alpha}{k}$, we need only know the positions where we know $\alpha$ is not zero, and then we can just accumulate the sums over these indices.
In the particular case of Python, this can be done very efficiently using JIT compilation (our implementation currently uses \texttt{Numba}).

We finally close by making a critical remark about the PGD step size, $\gamma_k$. Unlike the edge-based solver, PGD step size must be adjusted per commodity.
This is because the matrix $\atiter{A}{k}$ has very different dominant eigenvalues for each commodity (the largest eigenvalue is loosely dependent on the length of the longest path).
This consideration wasn't needed in the edge-based case, since the matrices we work with in that case have already been normalized.

 If $\Lambda_k$ is the largest eigenvalue of $\atiter{A}{k}$, one good choice is to let $\gamma_k \coloneqq \frac{1}{\Lambda_k}$. To calculate $\Lambda_k$ for large topologies, we use a few power method iterations (10 or 20 are enough for KDL).

\section{Adding Hierarchy}\label{app:adding-hierarchy}
Consider the following setting:
\begin{itemize}
    \item The network is partitioned into $\mathcal{D}$ domains. A variable is indexed for domain $d$ using brackets. For example, $X[d]$ denotes the matrix $X$ for domain $d$.
    \item The notation $X[!d]$ denotes sum over all partitions except $d$ (i.e.  $X[!d] \triangleq \sum_{d'} X[d'] - X[d]$.)
    \item Each domain will have one controller, which can freely and efficiently communicate with its set of switches.
    \item Each domain can \textit{only} manipulate its set of assignment variables, $X[d]$, and may only rely on estimates or aggregate values received from other domains.
\end{itemize}

We already have full decomposition over the demand/flow-conservation/non-negativity constraints.
The only thing that we must handle, is the capacity constraint and utilization objective.

One way to handle this is as follows. 
Define:
\begin{itemize}
    \item The domain aggregate flow as $Z[d] := \sum_k X[d]_k$.
    \item The pairing $R[d] = Z[d]$ (this allows us to pass information about other domain aggregates within loops).
    \item The domain MLU, $u[d]$ and its consensus variable $U$.
\end{itemize}
The original MLU problem can be rewritten as:
\begin{alignat*}{2}
    \text{minimize} 
        & \quad \frac{1}{\mathcal{D}} \sum_dg(U[d]) + 
        \sum_e f_e (\sum_d R[d]) + \\
        & \quad \sum_d \sum_e \sum_k f[d]_{ek} (X[d]_{ek}) \\
    \text{s.t.}
        & \quad R[!d] + Z[d] \preccurlyeq u[d]C && \forall d \\
        & \quad MX[d]_k = B[d]_k && \forall d, k \\
        & \quad 0 \preccurlyeq X[d]_k && \forall d, k \\
        & \quad 0 \leq U[d] && \forall d \\
        & \quad \sum_k X[d]_k - Z[d] = 0 && \forall d \\
        & \quad Z[d] - R[d] = 0 &&\forall d \\
        & \quad U[d] - O = 0 && \forall d
\end{alignat*}
The most important part of this rewrite is in the capacity constraint.
Exchanging $Z[d]$ with $R[d]$ for any domain except the \textit{current} one means that we can pass estimates of other domain aggregates without having to optimize them.
The price that we pay for this however, is that the capacity constraint now exists over every domain \textit{and} across anything that solves for $R[d]$ or $Z[d]$.

We can attack the last two constraints with ADMM.
Introduce the dual variables $l[d]$ and $v[d]$ in turn for these two constraints.

We first optimize over $X[d]$, $Z[d]$ and $U[d]$ by doing:
\begin{alignat*}{2}
    \text{minimize} 
        & \quad \frac{1}{\mathcal{D}} \sum_dg(U[d]) + 
        \sum_d \sum_e \sum_k f[d]_{ek} &&\Big(X[d]_{ek}\Big) \\
        & \quad \sum_d \frac{\beta}{2} \Big(U[d] - \atiter{O}{m} + \atiter{v[d]}{m}\Big)^2 + \\
        & \quad \sum_d \frac{\beta}{2} \Big\| Z[d] - \atiter{R[d]}{m} + \atiter{l[d]}{m} \Big\|_2^2 \\
    \text{s.t.}
        & \quad \atiter{R[!d]}{m} + Z[d] \preccurlyeq U[d]C && \forall d \\
        & \quad MX[d]_k = B[d]_k && \forall d, k \\
        & \quad \sum_k X[d]_k - Z[d] = 0 && \forall d \\
        & \quad 0 \leq U[d] && \forall d \\
        & \quad 0 \preccurlyeq X[d]_k && \forall d, k
\end{alignat*}
Crucially, this problem completely separates over domains (everything is indexed with $d$, including objective and constraints).
This gives us the \textit{Domain Optimization} problem (we can do away with domain indices to see the problem better).
\begin{alignat*}{2}
    \text{minimize} 
        & \quad \frac{1}{\mathcal{D}} g(U) + 
        \sum_e \sum_k f_{ek} (X_{ek}) \\
        & \quad + \frac{\beta}{2} \Big(U - \atiter{O}{m} + \atiter{v}{m}\Big)^2 \\
        & \quad + \frac{\beta}{2} \Big\| Z - \atiter{R}{m} + \atiter{l}{m} \Big\|_2^2 \\
    \text{s.t.}
        & \quad \atiter{R[!]}{m} + Z \preccurlyeq UC \\
        & \quad MX_k = B_k && \forall k \\
        & \quad 0 \preccurlyeq X_k && \forall k \\
        & \quad 0 \leq U \\
        & \quad \sum_k X_k - Z = 0
\end{alignat*}
This is pretty much the synchronous problem that we aimed to solve previously,
except here, we can get away with optimizing only copies of our own assignment, leaving only some headroom for other domains.
This problem we know very well how to solve efficiently.

The new part now, is the \textit{Domain Consensus} problem over $R[d]$ and $O$, which takes the form:
\begin{alignat*}{2}
    \text{minimize} 
        & \quad \sum_e f_e \Big(\sum_d R[d]\Big) + \\
        & \quad \sum_d \frac{\beta}{2} \Big(\atiter{U[d]}{m+1} - O + \atiter{v[d]}{m}\Big)^2 + \\
        & \quad \sum_d \frac{\beta}{2} \Big\| \atiter{Z[d]}{m+1} - R[d] + \atiter{l[d]}{m} \Big\|_2^2 \\
    \text{s.t.}
        & \quad R[!d] + \atiter{Z[d]}{m+1} \preccurlyeq \atiter{U[d]}{m+1}C && \forall d 
\end{alignat*}
The problem is unconstrained over $O$, admitting a simple closed form solution of:
\begin{alignat*}{2}
    \atiter{O}{m+1} \gets \atiter{\overline{U[d]}}{m+1} + \atiter{\overline{v[d]}}{m}
\end{alignat*}
And simplifying the remaining problem to:
\begin{alignat*}{2}
    \text{minimize} 
        & \quad \sum_e f_e \Big(\sum_d R[d]\Big) +\\
        & \quad \sum_d \frac{\beta}{2} \Big\| \atiter{Z[d]}{m+1} - R[d] + \atiter{l[d]}{m} \Big\|_2^2 \\
    \text{s.t.}
        & \quad R[!d] \preccurlyeq \atiter{U[d]}{m+1}C - \atiter{Z[d]}{m+1} \quad\forall d 
\end{alignat*}
This problem isn't too big, thus we can afford to add one new constraint to make it a bit more pleasant to look at.
Define $F := \sum_d R[d]$ to be the total reserved flow, then:
\begin{alignat*}{2}
    \text{minimize} 
        & \quad \sum_e f_e (F) +\sum_d \frac{\beta}{2} \Big\| \atiter{Z[d]}{m+1} - R[d] + \atiter{l[d]}{m} \Big\|_2^2 \\
    \text{s.t.}
        & \quad F - R[d] \preccurlyeq \atiter{U[d]}{m+1}C - \atiter{Z[d]}{m+1} \quad\forall d \\
        & \quad F - \sum_d R[d] = 0
\end{alignat*}
This is then finally finished with a dual update:
\begin{alignat*}{2}
    \atiter{v[d]}{m+1} &= \atiter{v[d]}{m} + \Big( \atiter{u[d]}{m+1} - \atiter{U}{m+1} \Big) \\
    \atiter{l[d]}{m+1} &= \atiter{l[d]}{m} + \Big( \atiter{Z[d]}{m+1} - \atiter{R[d]}{m+1} \Big)
\end{alignat*}
We now take a moment to recognize that this solution meets our requirements:
\begin{itemize}
    \item Assignments decompose over domains, and only demands specific to that domain are required.
    \item Each domain can optimize for utilization based on its own demands. Since this part never leaves that domain, we can reconvene in the domain controller as much as we like.
    \item The more the domains optimize, the better! The Domain Consensus problem benefits from it.
\end{itemize}
To some degree, this method is close in spirit to \cite{ncflow}, with the crucial difference that we do not heuristically break the capacity constraint into $\mathcal{D}$ inequalities for each domain and instead cut the link capacities with the same factor.

\begin{figure*}[h]
    \centering
    \begin{minipage}{0.28\textwidth}
        \centering
        \subfigure{\includegraphics[width=0.9\linewidth]{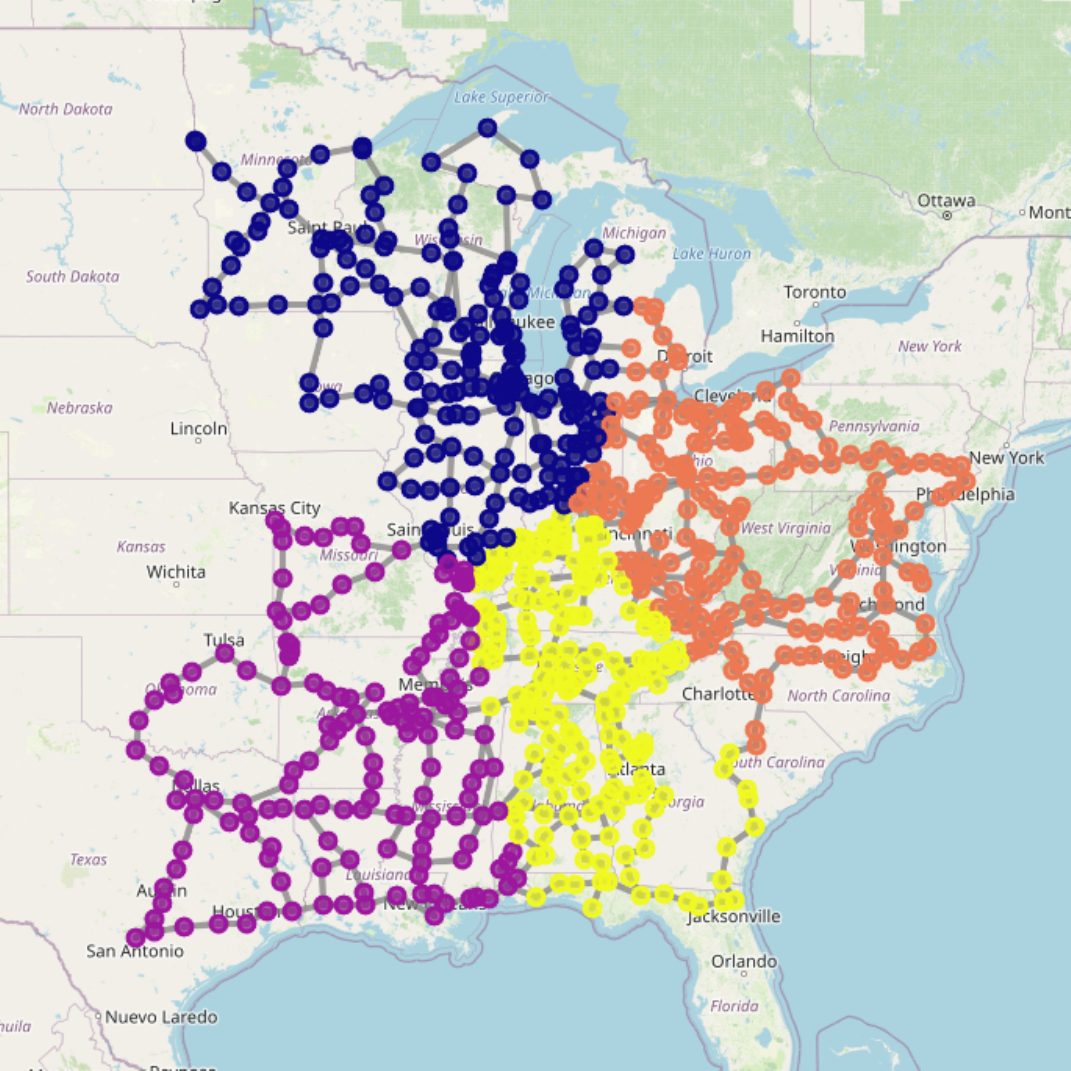}\label{fig:kdl}}
    \end{minipage}%
    \begin{minipage}{0.68\textwidth}
        \centering
        \subfigure{\includegraphics[width=0.9\linewidth]{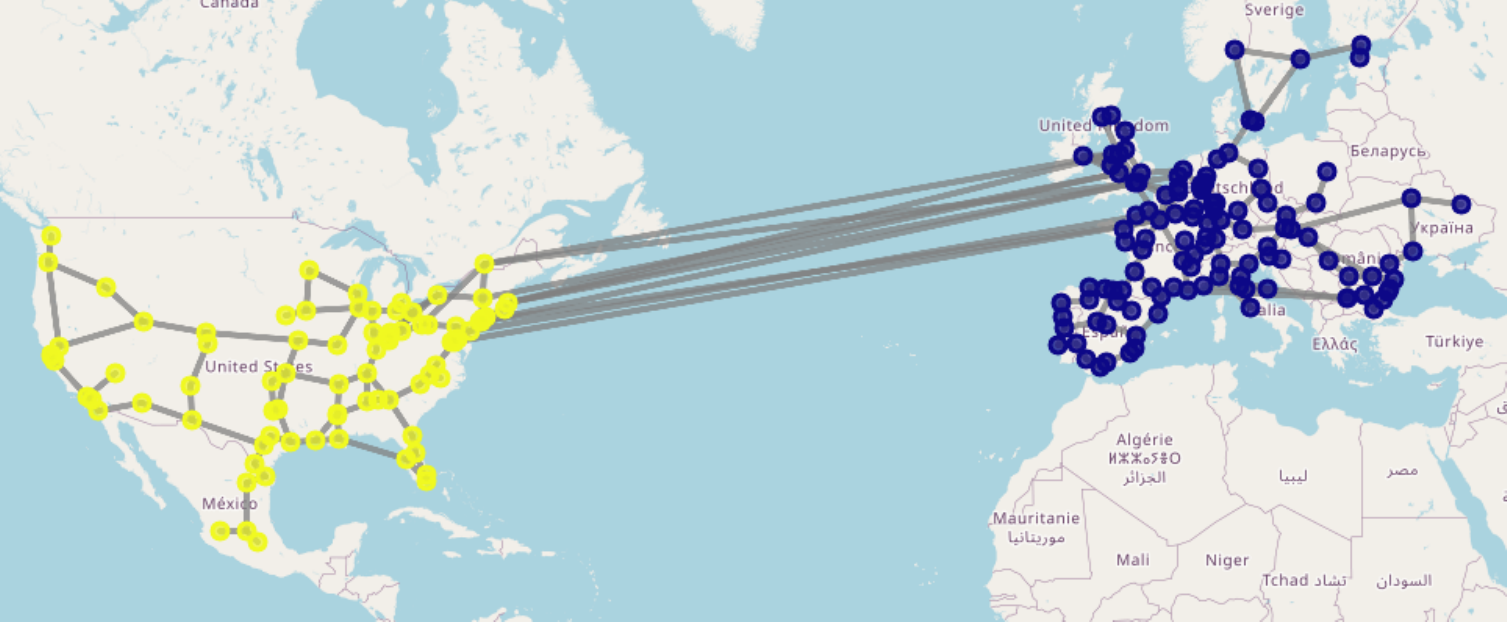}\label{fig:cogentco}}
    \end{minipage}
    \caption{We use two WANs in our evaluation: KDL (left) has 754 nodes and 1790 edges across the eastern US. Cogentco (right) has 190 nodes and 486 links spread across two continents.}
    \label{fig:topologies}
\end{figure*}


\end{document}